\documentclass[11pt]{article}

\usepackage{amsmath}
\usepackage{amssymb}
\usepackage{amsthm}
\usepackage{color}
\usepackage{soul}
\usepackage{mdwlist}
\usepackage{times}
\usepackage{mdframed}
\usepackage{paralist}
\usepackage{url}
\usepackage{xspace}
\usepackage{floatflt}
\usepackage{wrapfig}
\usepackage{cite}
\usepackage[hmargin=0.9in, vmargin=0.9in]{geometry}

\newtheorem{mainth}{Theorem}[section]

\newtheorem{Ndefinition}[mainth]{Definition}

\newtheorem{Nlemma}[mainth]{Lemma}

\newcommand{\blue}[1]{\textcolor{black}{{#1}}}

\newcommand{\newstr}[1]{}
\newcommand{\str}[1]{}
\newcommand{\remove}[1]{} 

\newcommand{\unit}{$\log_2|\mathcal{V}|$ bits}
\newcommand{\units}{$\log_2|\mathcal{V}|$ bits~}
\newcommand{\pre}{\text{`}\mathrm{pre}\text{'}}
\newcommand{\fin}{\text{`}\mathrm{fin}\text{'}}
\newcommand{\Null}{\text{`}\mathrm{null}\text{'}}

\newcommand{\gc}{\text{`}\mathrm{gc}\text{'}}

\begin{document}

\title{A Coded Shared Atomic Memory Algorithm for Message Passing Architectures}
\author{Viveck R. Cadambe \\
	RLE, MIT,\\
	Cambridge, MA, USA\\
	{\tt viveck@mit.edu}
	\and
Nancy Lynch\\
CSAIL, MIT \\
{Cambridge, MA, USA}\\
	{\tt lynch@theory.lcs.mit.edu}
\and
Muriel M\'{e}dard \\
RLE, MIT\\
{Cambridge, MA, USA}\\
	{\tt medard@mit.edu}
	\and
Peter Musial\\
Advanced Storage Division, EMC$^{2}$\\
{Cambridge, MA, USA}\\
{\tt peter.musial@emc.com}
\thanks{This work was supported by in part by AFOSR contract no. FA9550-13-1-0042, NSF award no.s CCF-1217506 and 0939370-CCF, and by BAE Systems National Security Solutions, Inc., award 739532-SLIN 0004.}
}

\date{}
    
\maketitle
\begin{abstract}
This paper considers the communication and storage costs of emulating
atomic (linearizable) multi-writer multi-reader shared memory in distributed
message-passing systems. 
%
The paper contains three main contributions:

\noindent{\bf (1)}
We present a atomic shared-memory emulation algorithm that we call \emph{Coded Atomic Storage} (CAS).  This algorithm uses \emph{erasure coding} methods. In a storage system with $N$ servers that is resilient
to $f$ server failures, {we show that} the communication cost of CAS is
$\frac{N}{N-2f}$. The storage cost of CAS is unbounded. 

\noindent{\bf (2)}
We present a modification of the CAS algorithm known as CAS with Garbage Collection (CASGC). The CASGC algorithm is parametrized by an integer $\delta$ and has a bounded storage cost. We show that in every execution where the number of write operations that are concurrent with a read operation is no bigger than $\delta$, the CASGC algorithm with parameter $\delta$ satisfies atomicity and liveness. We explicitly characterize the storage cost of CASGC, and show that it has the same communication cost as CAS.

\noindent{\bf (3)}
We describe an algorithm known as the Communication Cost Optimal Atomic Storage (CCOAS) algorithm that achieves a smaller communication cost than CAS and CASGC. In particular, CCOAS incurs read and write communication costs of $\frac{N}{N-f}$ measured in terms of number of object values. We also discuss drawbacks of CCOAS as compared with CAS and CASGC.
\end{abstract}

\renewcommand{\thefootnote}{\roman{footnote}}

\newpage

\section{Introduction} 
\label{sec:Intro}
Since the late 1970s, emulation of shared-memory systems in
distributed message-passing environments has been an active area of
research~\cite{Gifford79,Tho79,ABD,MR98,RAMBO,DynaStore,FLS01, LDR, AJ95, DGL,FAB,AJX,GWGR,HGR,CT, Anderson_etal, dobre_powerstore}.
The traditional approach to building redundancy for distributed
systems in the context of shared memory emulation is \emph{replication}.
%
In their seminal paper \cite{ABD}, Attiya, Bar-Noy, and Dolev presented a replication based algorithm for
emulating shared memory that achieves atomic
consistency~\cite{Lamport86,Herlihy90}. In this paper we consider a simple multi-writer generalization of their algorithm which we call the \emph{ABD} algorithm\footnote{{The algorithm of Attiya, Bar-Noy and Dolev \cite{ABD} allows only a single node to act as a writer.  Also, it did not distinguish between client and server nodes as we do in our paper.}}.
This algorithm uses a quorum-based replication scheme~\cite{Vukolic12},
combined with read and write protocols to ensure that the emulated object is {atomic {\cite{Lamport86}} (linearizable {\cite{Herlihy90}}), and to ensure liveness, specifically, that each operation terminates provided that at most $\lceil \frac{N-1}{2}\rceil$
server nodes fail.}
A critical step in ensuring atomicity in ABD is the
\emph{propagate} phase of the read protocol, where the readers write back the value they read to a subset of the server nodes.
Since the read and write protocols require multiple communication
phases where entire replicas are sent, this algorithm has a high
communication cost. In \cite{LDR}, Fan and Lynch introduced a directory-based replication
algorithm known as the LDR algorithm that, like \cite{ABD}, emulates atomic shared memory in the
message-passing model;  however, unlike \cite{ABD}, its read protocol 
is required to write only some metadata information to the
directory, rather than the value read.
In applications where the data being replicated is much larger than
the metadata, LDR is less costly than 
ABD in terms of communication costs.

The main goal of our paper is to develop shared memory emulation algorithms, based on the idea of \emph{erasure coding}, that are efficient in terms of communication and storage costs. Erasure coding is a generalization of replication that is well known in the context of classical storage
systems~\cite{LinCostello_Book,Roth_CodingTheory,oggier, cassuto}.
Specifically, in erasure coding, each server does not store the value
in its entirety, but only a part of the value called a \emph{coded
element}. In the classical coding theory framework which studies storage of a single version of a data object, this approach is well known to lead to smaller 
storage costs as compared to replication (see Section \ref{sec:ec}). Algorithms for shared memory emulation that use the idea of erasure coding to store multiple versions of a data object consistently have been developed in \cite{AJ95, DGL,FAB,AJX,GWGR,HGR,CT, Anderson_etal, dobre_powerstore}.
In this paper, we develop algorithms that improve on previous algorithms in terms of communication and storage costs. We summarize our main contributions and compare them with previous related work next.

\subsection*{Contributions}
We consider a static distributed message-passing
setting where the universe of nodes is fixed and known, and nodes
communicate using a reliable message-passing network. We assume that 
client and server nodes can fail. {We define our system model, and communication and storage cost measures in 
Sec.~\ref{sec:background}.} 

\emph{The CAS algorithm:} We develop the \emph{Coded Atomic Storage} (CAS)
algorithm presented in Section~\ref{sec:CAS}, which is an erasure coding based shared memory emulation algorithm. {We present a brief introduction of the technique of erasure coding in Section~\ref{sec:ec}.} 
For a storage system with $N$ nodes, we show in Theorem \ref{lem:liveness} that CAS ensures the following liveness property: all operations that are invoked by a non-failed client terminate provided
that the number of \emph{server} failures is bounded by a parameter $f,$ where $f < \lceil \frac{N}{2}\rceil$ and regardless of the number of client failures. We also show in Lemma \ref{lem:liveness} that CAS ensures atomicity
regardless of the number of {(client or server)} failures. In Theorem \ref{thm:comCAS} in Section \ref{sec:CAS}, we also analyze the communication cost of CAS. 
Specifically, in a storage system with $N$ servers that is resilient to $f$ server node failures, we show that the communication costs of CAS are equal to $\frac{N}{N-2f}$. We note that these communication costs of CAS are smaller than replication based schemes ({see Appendix \ref{app:ABD_LDR}} for an analysis of communication costs of ABD and LDR algorithms.).  The storage cost of CAS, however, are unbounded because each server stores the value associated with the latest version of the data object it receives. Note that in comparison, in the ABD algorithm which is based on replication, the storage cost is bounded because each node stores only the latest version of the data object (see Appendix \ref{app:ABD_LDR} for an explicit characterization of the storage cost incurred by ABD). 



\emph{The CASGC algorithm:} In Section \ref{sec:CASgc}, we present a variant of
CAS called the CAS with Garbage Collection (CASGC) algorithm, which achieves a
bounded storage cost by \emph{garbage collection}, i.e., {discarding} values
associated with sufficiently old versions. CASGC is parametrized by an integer $\delta$ which, informally speaking, controls the number of tuples that each server stores. We show that CASGC satisfies atomicity in Theorem \ref{thm:CASGCatomicity} by establishing a {formal simulation relation \cite{Lynch1996}} between CAS and CASGC.  Because of the garbage collection at the servers, the liveness conditions for CASGC are more stringent than CAS. The liveness property satisfied by CASGC is described in Theorem \ref{thm:CASGCliveness} in Section \ref{sec:CASgc}, where we argue that in an execution of CASGC where the number of write
operations concurrent with a read operation is no bigger than a parameter $\delta$, every operation terminates.
 The main technical challenge lies in careful design of the CASGC algorithm in order to ensure that an unbounded number of writes that fail before propagating enough number of coded elements do not prevent a future read from returning a value of the data object. In particular, failed writes that begin and end before a read is invoked are not treated as operations that are concurrent with the read, and therefore do not contribute to the concurrency limit of $\delta$. While CASGC incurs the same communication costs as CAS, it incurs a bounded storage cost. A non-trivial bound on the storage cost incurred by an execution of CASGC is described in Theorem \ref{thm:storCAS}.  


\emph{Communication Cost Lower Bound:}
In Section \ref{sec:lowerbound} we describe a new algorithm called the Communication Cost Optimal Atomic Storage (CCOAS) algorithm that satisfies the same correctness conditions as CAS, but incurs smaller communication costs. However, CCOAS {would not be easily generalizable} to settings where channels could incur losses because, {unlike CAS and CASGC, it requires that messages from clients to servers are delivered reliably even after operations associated with the message terminates. While CCOAS is applicable in our model of reliable channels, designing a protocol with this property may not be possible when the channel has losses especially if the client fails before delivering the messages.} We describe CCOAS, analyse its communication costs, and discuss its drawbacks in Section \ref{sec:lowerbound}.

\subsection*{Comparison with Related Work}
Erasure coding has been used to develop shared memory emulation techniques for systems with crash failures in \cite{AJ95, DGL,FAB,AJX} and {Byzantine} failures in \cite{GWGR,HGR,CT, dobre_powerstore}\footnote{An earlier version of our work is presented in a technical report \cite{Cadambe_Lynch_Medard_Musial_TR}.}. In erasure coding, note that each server stores a coded element, so a reader has to obtain enough coded elements to decode and return the value. The main challenge in extending replication based algorithms such as ABD to erasure coding lies in handling partially completed or failed writes. In replication, when a read occurs during a partially completed write, servers simply send the stored value and the reader returns the latest value obtained from the servers. However, in erasure coding, the challenge is to ensure that a read that observes the trace of a partially completed or failed write obtains a enough coded elements corresponding to the same version to return a value. Different algorithms have different approaches in handling this challenge of ensuring that the reader decodes a value of the data object. As a consequence, the algorithms differ in the liveness properties satisfied, and the communication and storage costs incurred. We discuss the differences here briefly. 

Among the previous works, \cite{DGL,HGR, CT, dobre_powerstore} have similar correctness requirements as our paper; these references aim to emulate an atomic shared memory that supports concurrent operations in asynchronous networks. We note that the algorithm of \cite{CT} cannot be generalized to lossy channel models (see discussion in \cite{DGL}). We compare our algorithms with the \emph{ORCAS-B} algorithm of \cite{DGL}\footnote{The \emph{ORCAS-A} algorithm of \cite{DGL}, although uses erasure coding, has the same \emph{worst case} communication and storage costs as ABD.}, the algorithm of \cite{HGR}, which we call the \emph{HGR} algorithm, and the \emph{M-PoWerStore} algorithm of \cite{dobre_powerstore}. We note that \cite{DGL} assumes lossy channels and \cite{HGR, dobre_powerstore} assume Byzantine failures. Here, we interpret the algorithms of \cite{DGL,HGR, dobre_powerstore} in our model that has lossless channels and crash failures, and use worst-case costs for comparison. 

The CAS and CASGC algorithms resemble the M-PoWerStore and HGR algorithms in their structure. These algorithms handle partially completed or failed writes by \emph{hiding} ongoing writes from a read until enough coded elements have been propagated to the servers. The write communication costs of CAS, CASGC, M-PoWerStore, HGR and ORCAS-B are all the same. However, there are differences between these algorithms in the liveness properties, garbage collection strategies and read communication costs. 

CAS is essentially a {restricted version of the \emph{M-PoWerStore} algorithm of \cite{dobre_powerstore} for the crash failure model.} The main difference between CAS and M-PoWerStore is that in CAS, servers perform gossip\footnote{As we shall see later, the server gossip is not essential to correctness of CAS. It is however useful as a theoretical tool to prove correctness of CASGC.}. However, M-PoWerStore does not involve garbage collection and therefore incurs an infinite storage cost. {The garbage collection strategies of HGR and ORCAS-B are similar to that of CASGC with the parameter $\delta$ set to $1$. In fact, the garbage collection strategy of CASGC may be viewed as a non-trivial generalization of the garbage collection strategies of ORCAS-B and HGR. We next discuss differences between these algorithms in terms of their liveness properties and communication costs.}

The ORCAS-B algorithm satisfies the same liveness properties as ABD and CAS, which are stronger than the liveness conditions of CASGC. However, in ORCAS-B, to handle partially completed writes, a server sends coded elements corresponding to multiple versions to the reader. This is because, in ORCAS-B, a server, on receiving a request from a reader, registers the client\footnote{{The idea of registering a client's identity was introduced originally in \cite{Martin_minimal} and plays an important role in our CCOAS algorithm as well.}} and sends all the incoming coded elements to the reader until the read receives a second message from a client. Therefore, the read communication cost of ORCAS-B grows with the number of writes that are concurrent with a read. In fact, in ORCAS-B, if a read client fails in the middle of a read operation, servers may send all the coded elements it receives from future writes to the reader. In contrast, CAS and CASGC have smaller communication costs because each server sends only one coded element to a client per read operation, irrespective of the number of writes that are concurrent with the read.

In HGR, read operations satisfy \emph{obstruction freedom}, {that is, a read returns if there is a period during the read where no other operation takes steps for sufficiently long.  Therefore, in HGR, operations may terminate even if the number of writes concurrent with a read is arbitrarily large, but it requires a sufficiently long period where concurrent operations do not take steps. On the contrary, in CASGC, by setting $\delta$ to be bigger than $1$, we ensure that read operations terminate even if concurrent operations take steps, albeit at a larger storage cost, so long as the number of writes concurrent with a read is bounded by $\delta$. Interestingly, the read communication cost of HGR is larger than CASGC, and increases with the number of writes concurrent to the read to allow for read termination in presence of a large number of concurrent writes.}

We note that the server protocol of the CASGC algorithm is more complicated as compared with previous algorithms.  In particular, unlike ORCAS-B, HGR and M-PoWerStore, the CASGC algorithm requires gossip among the servers to ensure read termination in presence of concurrent writes at a bounded storage cost and low communication cost. A distinguishing feature of our work is that we provide formal measures of communication and storage costs of our algorithms. Our contributions also include the CCOAS algorithm, and complete correctness proofs of all our algorithms through the {development of invariants and simulation relations}, which may be of independent interest. The generalization of CAS and CASGC algorithms to the models of \cite{HGR, DGL, dobre_powerstore, CT} which consider Byzantine failures and lossy channel models is an interesting direction for future research.

\section{System Model}\label{sec:background}

\subsection{Deployment setting.}
We assume a {\em static asynchronous deployment setting} where 
all the nodes and the network connections are known a priori and the only sources of 
dynamic behavior are node stop-failures (or simply, failures) and processing and communication delays.  
We consider  a {message-passing} setting where nodes communicate via point-to-point reliable channels. We assume a universe of nodes that is {the} union of {\em server} and {\em client} nodes, where the
{client} nodes are {\em reader} or {\em writer} nodes. $\mathcal{N}$ represents the set of 
server nodes; $N$ denotes the cardinality of $\mathcal{N}.$ We assume that server and client 
nodes can fail (stop execution) at any point.  We assume that the number of server node failures 
is at most $f$. There is no bound on the number of client failures.



\subsection{Shared memory emulation.}
We consider algorithms that emulate multi-writer, multi-reader (MWMR) read$/$write atomic shared memory using our deployment platform.  We assume that read clients receive read requests (invocations) from some local external source, and respond with object values.  Write clients receive write requests and respond with acknowledgments.
The requests follow a ``handshake'' discipline, where a new invocation at a client waits for a response to the preceding invocation at the same client. We require that the overall external behavior of the algorithm corresponds to atomic {(linearizable) }memory. For simplicity, in this paper we consider a shared-memory system that consists of just a single object.

We represent each version of the data object as {a} $(tag,value)$ pair.
When a write client processes a write request, it assigns a \emph{tag} to the request. We assume that the tag is an element of a totally ordered set $\mathcal{T}$ that has a minimum element $t_0$. 
The tag of a write request serves as a unique identifier for that request, and 
the tags associated with successive write requests at a particular write client increase monotonically. 
We assume that $value$ is a member of a finite set $\mathcal{V}$ that represents the set of values that the data object can take on; note that $value$ can be represented by $ \log_{2}|\mathcal{V}|$ bits\footnote{{Strictly speaking, we need $\lceil \log_{2}|\mathcal{V}|\rceil$ bits since the number of bits has to be an integer. We ignore this rounding error.}}. We assume that all servers are initialized with a default initial state.

\subsection{Requirements} 
The key correctness requirement on the targeted shared memory service is 
{\em atomicity.}
A shared atomic object is one that supports concurrent access by multiple clients 
and where the {observed global external behaviors {``look like''} the object is being 
accessed sequentially}.
Another requirement is \emph{liveness}, by which we mean {here} that an operation of a 
non-failed client is guaranteed to terminate provided that the number of server failures 
is at most $f$, and irrespective of the failures of other clients\footnote{We assume 
that $N > 2f,$ since correctness cannot be guaranteed if $N \leq 2f$ \cite{Lynch1996}.}.

\subsection{Communication cost}
Informally speaking, the communication cost is the number of bits transferred over the 
{point-to-point links in }the message-passing system.
{For a message that can take any value in some finite set $\mathcal{M}$, we measure its communication cost as} 
$\log_{2}|\mathcal{M}|$ bits. We separate the cost of communicating {a value of }the data object from the cost of communicating the tags and other metadata. Specifically, we assume that each message is a triple $(t,w,d)$ where $t \in \mathcal{T}$ is a tag, $w \in \mathcal{W}$ is a {component} of the triple that depends on the value associated with tag $t$, and $d \in \mathcal{D}$ is any additional metadata that is independent of the value. Here, $\mathcal{W}$ is a finite set of values that {the second component} of the message can take {on}, depending on the value \str{to }{of }the data object. $\mathcal{D}$ is a finite set that contains all the possible metadata elements for the message. These sets are assumed to be known a priori to the sender and recipient of the message. In this paper, we make the approximation: $ \log_{2}|\mathcal{M}| \approx \log_{2}|\mathcal{W}|,$ that is, the costs of communicating the tags and the metadata are negligible as compared to 
the cost of communicating the data object values. We assume that every message is sent on behalf of some read or write operation. We next define the read and write communication costs of an algorithm.

{ For a given shared memory algorithm, consider an execution $\alpha$. The communication cost of a write operation in $\alpha$ is the sum of the communication costs of all the messages sent over the point-to-point links on behalf of the operation. The write communication cost of the execution $\alpha$ is the supremum of the costs of all the write operations in $\alpha$. The write communication cost of the algorithm is the supremum of the write communication costs taken over all executions. }The read communication cost of an algorithm is defined similarly.

\subsection{Storage cost}
Informally speaking, at any point of an execution of an algorithm, the {\em storage cost} is the total number of bits stored by the servers. 
Specifically, we assume that a server node stores a set of triples with each triple of the form $(t, w, d)$, where  $t \in \mathcal{T},$ $w$ depends on the value of the data object associated with tag $t$, and $d$ represents additional metadata that is independent of the values stored. We neglect the cost of storing the tags and the metadata; so the cost of storing the triple $(t,w,d)$ is measured as $\log_2|\mathcal{W}|$ bits. The storage cost of a server is the sum of the storage costs of all the triples stored at the server.  {For a given shared memory algorithm, consider an execution $\alpha$. The storage cost at a particular point of $\alpha$ is the sum of the storage costs of all the non-failed servers at that point. The storage cost of the execution $\alpha$ is the supremum of the storage costs over all points of $\alpha$. The storage cost of an algorithm is the supremum of the storage costs over all executions of the algorithm}. 

\section{Erasure Coding - Background}\label{sec:ec}

Erasure coding is a generalization of replication that has been widely studied for purposes of failure-tolerance in storage systems (see \cite{LinCostello_Book,Roth_CodingTheory,plank2005t1,oggier,cassuto}). The key idea of erasure coding involves splitting the data into several \emph{{coded elements}}, each of which is stored at a different server node. As long as a sufficient number of coded elements can be accessed, the original data can be recovered.  Informally speaking, given two positive integers $m, k,$ {$k < m$}, \emph{an $(m,k)$ Maximum Distance Separable (MDS) code maps a $k$-length vector to an $m$-length vector, where the input $k$-length vector can be recovered from any $k$ coordinates of the output $m$-length vector.} This implies that an $(m,k)$ code, when used to store a $k$-length vector on $m$ server nodes - each server node storing one of the $m$ coordinates of the output - can tolerate $(m-k)$ node failures in the absence of any consistency requirements (for example, see \cite{RAID}). We proceed to define the notion of an MDS code formally.


  Given an arbitrary finite set $\mathcal{A}$ and any set $S \subseteq \{1,2,\ldots,m\},$ let $\pi_{S}$ denote the \emph{natural projection mapping} from $\mathcal{A}^{m}$ 
onto the coordinates corresponding to $S,$ i.e., denoting $S=\{s_{1},s_{2},\ldots,s_{|S|}\},$ where {$s_1 < s_2 \ldots < s_{|S|}$,} the 
function $\pi_{S}:\mathcal{A}^{m}\rightarrow \mathcal{A}^{|S|}$ is defined as
$\pi_{S}\left(x_1,x_2,\ldots,x_m\right) = (x_{s_1},x_{s_2},\ldots, x_{s_{|S|}})$.
\begin{Ndefinition}[Maximum Distance Separable (MDS) code] Let $\mathcal{A}$ denote any finite set.  For positive integers $k, m$ such that $k < m,$ an $(m,k)$ code over $\mathcal{A}$ \str{consists of }{is }a map $\Phi:\mathcal{A}^{k} \rightarrow \mathcal{A}^{m}$.  An $(m,k)$ code {$\Phi$} over $\mathcal{A}$ is said to be \emph{Maximum Distance Separable} (MDS) if, for every $S \subseteq \{1,2,\ldots,m\}$ where $|S|=k,$ there exists a function $\Phi^{-1}_{S}:\mathcal{A}^{k} \rightarrow \mathcal{A}^{k}$ such that: $\Phi_S^{-1}(\pi_{S}(\Phi(\mathbf{x})) = \mathbf{x}$  for every $\mathbf{x} \in \mathcal{A}^{k}$, where $\pi_{S}$ is the natural projection mapping.
\label{def:MDS}
\end{Ndefinition}
We refer to each of the $m$ coordinates of the output of {an} $(m,k)$ code $\Phi$ as a \emph{coded element}.  Classical $m$-way replication, where the input value is repeated $m$ times, is in fact {an} $(m,1)$ MDS code{.} Another example is the \emph{single parity code}: an $(m,m-1)$ MDS code over $\mathcal{A}=\{0,1\}$ which maps the $(m-1)$-bit vector $x_1, x_2,\ldots,x_{m-1}$ to the $m$-bit vector $x_1,x_2,\ldots,x_{m-1},x_{1}\oplus x_{2}\oplus \ldots \oplus x_{m-1}.$ 

%
{
We now review the use of an MDS code in the classical coding-theoretic model, where a single version of a data object with value $v \in \mathcal{V}$ is stored over $N$ servers using an $(N,k)$ MDS code. We assume that $\mathcal{V}= \mathcal{W}^{k}$ for some finite set $\mathcal{W}$ and that an $(N,k)$ MDS code {$\Phi:\mathcal{W}^{k}\rightarrow \mathcal{W}^{N}$} exists over $\mathcal{W}$ (see Appendix \ref{app:coding} for a discussion). The value $v$ of the data object can be used as an input to {$\Phi$ }to get $N$ coded elements over $\mathcal{W};$ each of the $N$ servers, respectively, stores one of these coded elements. Since each coded element belongs to the set $\mathcal{W},$ whose cardinality satisfies 
$|\mathcal{W}| = |\mathcal{V}|^{1/k} = 2^{\frac{\log_{2}|\mathcal{V}|}{k}},$ \emph{each coded element can be represented as a $\frac{\log_{2}|\mathcal{V}|}{k}$ bit-vector, i.e., the number of bits in each coded element is a fraction $\frac{1}{k}$ of the number of bits in the original data object}. When we employ an $(N,k)$ code in the context of storing multiple versions, the size of a coded element is closely related to communication and storage costs incurred by our algorithms (see Theorems \ref{thm:comCAS} and \ref{thm:storCAS}).

\remove{We also describe a classical coding-theoretic cost measure known as the \emph{redundancy factor} and its connection to the storage costs.} 

\section{Coded Atomic Storage}\label{sec:CAS}
We now present the {\em Coded Atomic Storage} (CAS) {algorithm, which takes 
advantage of erasure coding techniques to reduce the communication \str{and 
storage }cost for emulating atomic shared memory. CAS is parameterized by an 
integer $k$, $1 \leq k \leq N-2f$; we denote the algorithm with parameter value 
$k$ by CAS($k$).}  CAS, like ABD and LDR, is a quorum-based algorithm. {Later, in Sec. \ref{sec:CASgc}, we present a variant of CAS that has efficient storage costs as well (in addition to having the same communication costs as CAS).}


\begin{figure*}[!t]
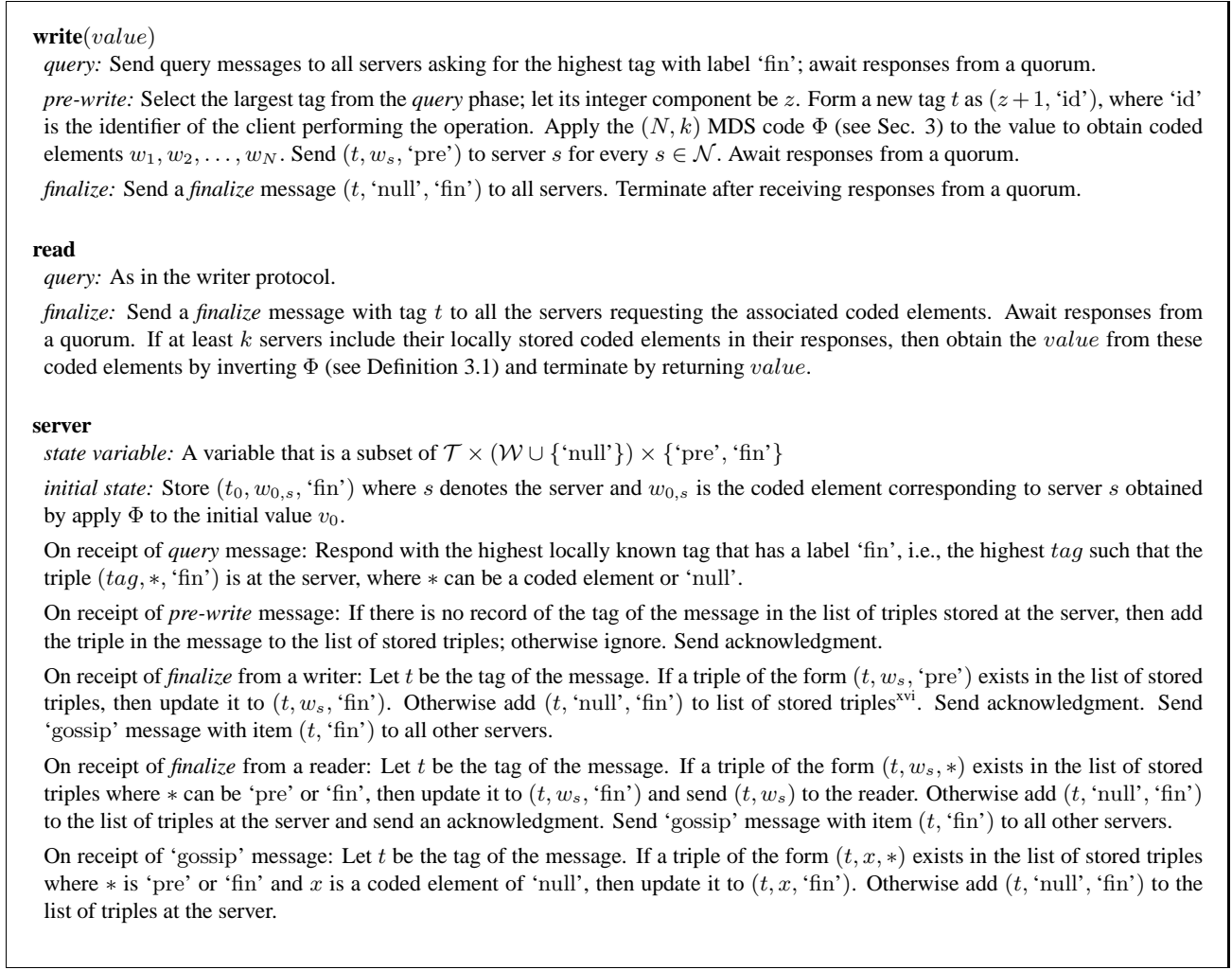

  \begin{mdframed}
  \footnotesize
  \begin{tabbing}
  	{\bf write$(value)$} \\
	\ \ \begin{minipage}[t]{\textwidth}%
	{\em query:}  Send query messages to all servers asking for the 
	highest tag with label $\fin$; {await responses from a quorum.} \smallskip\\
	{\em pre-write:} {Select the largest tag from the \emph{query} phase; let its integer component be $z$. Form a new tag $t$ as $(z+1,\text{`}\mathrm{id}\text{'})$, where `$\mathrm{id}$' is the identifier of the client performing the operation.} Apply the $(N,k)$ MDS code {$\Phi$} (see Sec. \ref{sec:ec}) {to }the value to obtain coded elements ${w_1, w_2, \ldots, w_N.}$ Send $(t, w_s,\pre)$ to server $s$ for every $s \in \mathcal{N}.$ Await response{s} from a quorum. \smallskip \\
	{\em finalize:} Send a \emph{finalize} message {$(t,\text{`}\mathrm{null}\text{'},\fin)$ to all
	servers.}  Terminate after receiving responses from a quorum.\smallskip
	\end{minipage}\\ \\
	{\bf read} \\
	\ \ \begin{minipage}[t]{\textwidth}%
	{\em query:} As in the writer protocol. \smallskip \\
	{\em finalize:} Send a \emph{finalize} message with tag $t$ {to all the servers requesting 
	the associated coded elements.}\str{ to the all servers. } Await responses 
	from a quorum. {If at least $k$ servers include their locally stored coded elements in their  
	responses, then obtain the $value$ from these coded elements by inverting $\Phi$ (see 
	Definition \ref{def:MDS}) and terminate by returning $value$.} \smallskip
	\end{minipage} \\ \\
	{\bf server} \\
	\ \ \begin{minipage}[t]{\textwidth}%
\emph{state variable:} A variable that is a subset of $\mathcal{T} \times \left(\mathcal{W} \cup \{\Null\}\right) \times \{\pre, \fin\}$\smallskip \\
		{{\em initial state:} Store {$(t_0, w_{0,s}, \fin)$ where $s$ denotes the server and $w_{0,s}$ is the coded element corresponding to server $s$ obtained by apply $\Phi$ to the initial value $v_0$.}} \smallskip  \\
		On receipt of {{\em query} message: Respond with the highest locally known tag that has a label $\fin$, i.e., the highest $tag$ such that the triple $(tag, *, \fin)$ is at the server, where $*$ can be a coded element or `$\mathrm{null}$'.} \smallskip \\
	On receipt of {{\em pre-write} message:} If there is no record of the tag of the message in the list of triples stored at the server, then {add }the \str{incoming }triple {in the message }to the list of stored triples; otherwise ignore. {Send acknowledgment.} \smallskip \\
	On receipt of {{\em finalize} from a writer: Let $t$ be the tag of the message. If a triple of the form $(t, w_s, \pre)$ exists in the list of stored triples, then update it to $(t, w_s, \fin)$. Otherwise add $(t, \text{`}\mathrm{null}\text{'}, \fin)$ to list of stored triples\footnotemark[16]. Send {acknowledgment}.} {Send `$\mathrm{gossip}$' message with item $(t,\fin)$ to all other servers.}  \smallskip \\
	On receipt of {\em finalize} from a reader: Let $t$ be the tag of the message. If a triple of the form $(t, w_s, *)$ exists in the list of stored triples where $*$ can be $\pre$ or $\fin$, then update it to $(t, w_s, \fin)$ and send $(t, w_s)$ to the reader. Otherwise add $(t, \text{`}\mathrm{null}\text{'}, \fin)$ to the list of triples at the server and send an {acknowledgment.} Send `$\mathrm{gossip}$' message with item $(t,\fin)$ to all other servers.  \smallskip \\
	On receipt of `$\mathrm{gossip}$' message: Let $t$ be the tag of the message. If a triple of the form $(t, x, *)$ exists in the list of stored triples where $*$ is $\pre$ or $\fin$ and $x$ is a coded element of $\Null$, then update it to $(t, x , \fin)$. Otherwise add $(t, \text{`}\mathrm{null}\text{'}, \fin)$ to the list of triples at the server. \smallskip \\
	\end{minipage}
	\end{tabbing}
  \end{mdframed}
  \caption{Write, read, {and server} protocols of the CAS algorithm.}\label{fig:CAS}
\end{figure*}

%
Handling of incomplete writes is not as simple when erasure coding is used because, unlike in replication based techniques, no {single} server has a complete replica of the value being written.
In CAS, we solve this problem by \emph{hiding} ongoing write operations from reads until enough information has been stored at servers{.} Our approach essentially mimics \cite{dobre_powerstore}, projected to the setting of crash failures. We describe CAS in detail next.

\noindent{\bf Quorum specification.}
%
We \str{construct }{define} our quorum system, $\mathcal{Q},$ to be the set of all subsets of 
$\mathcal{N}$ that have at least $\lceil \frac{N+k}{2}\rceil $ elements (server nodes).
We refer to the members of $\mathcal{Q}$, 
as quorum sets. We show in Apppendix \ref{appendix:sec:quorum} that $\mathcal{Q}$ satisfies the following property:

\begin{Nlemma}
\label{lem:quorums}
Suppose that $1 \leq k \leq N-2f.$ 
({\bf i}) {If $Q_{1}, Q_{2} \in \mathcal{Q},$ then} $|Q_{1} \cap Q_{2}| \ge k$. 
({\bf ii}) If the number of failed servers is at most $f$,
then {$\mathcal{Q}$ contains at least one quorum set $Q$ of non-failed servers.}
\end{Nlemma}

The CAS algorithm can, in fact, use any quorum system that satisfies properties ({\bf i}) and ({\bf ii}) of 
Lemma \ref{lem:quorums}.

\subsection{Algorithm description}
In CAS, we assume that tags are tuples of the form $(z, \text{`}\mathrm{id}\text{'})$, where $z$ is an integer and $\text{`}\mathrm{id}\text{'}$ is an identifier of a client node. The ordering on the set of tags $\mathcal{T}$ is defined {lexicographically, using the usual ordering on the integers and a predefined ordering on the client identifiers.}
{We add a `$\mathrm{gossip}$' protocol to CAS, whereby each server sends each \emph{item} from $\mathcal{T}\times \{\fin\}$ that it ever receives once (immediately) to every other server. As a consequence, in any fair execution, if a non-failed server initiates `$\mathrm{gossip}$' or receives `$\mathrm{gossip}$' message with item $(t,\fin)$, then, every non-failed server receives a `$\mathrm{gossip}$' message with this item at some point of the execution.} Fig.~\ref{fig:CAS} contains a description of the read and write protocols,  
and the server actions of CAS. Here, we provide an overview of the algorithm.

Each server node maintains a set 
of $(tag, coded\text{-}element, label)$\footnote{The `$\mathrm{null}$' entry indicates that no coded element is stored; the storage cost associated storing a $\mathrm{null}$ coded element is negligible.} triples, where we specialize the metadata to 
$label\in \{\pre,\fin\}$.  The different phases of the write and read protocols are executed sequentially.  In each phase, a client sends messages to servers to which the non-failed servers respond. Termination of each phase depends on getting responses from at least one quorum.


The {\em query} phase is identical in both protocols and it allows clients to discover a recent \emph{finalized 
object version}{, i.e., a recent version with a $\fin$ tag.}
The goal of the {\em pre-write} phase {of a write} is to ensure that each server gets a tag and a coded element with label $\pre$.  Tags associated with label $\pre$ are not 
visible to the readers, since the servers respond to \emph{query} messages {only }with finalized tags. Once a quorum, say $Q_{pw},$ has acknowledged receipt of the coded elements to the pre-write phase,
the writer proceeds to its \emph{finalize} phase. In this phase, it propagates a finalize ($\fin$) label with the tag and waits for a response from a quorum of servers, say $Q_{fw}.$ The purpose of propagating the $\fin$ label is to record that the coded elements associated with the tag have been propagated to a quorum\footnote{It is worth noting that $Q_{fw}$ and $Q_{pw}$ need not be the same quorum.}. In fact, when a tag {appears anywhere in the system associated with a $\fin$ label}, it means that the corresponding coded elements\str{ have} reached a quorum $Q_{pw}$ with a $\pre$ label at some previous point. The operation {of a writer} in {the }two phases {following its \emph{query phase }}helps overcome the challenge of handling writer failures. In particular, notice that only tags with the $\fin$ label are visible to the 
reader. This ensures that the reader gets at least $k$ unique coded elements from any quorum of non-failed nodes in response to its finalize messages, because such a quorum has an intersection of at least $k$ nodes with $Q_{pw}$. Finally, the reader helps propagate the tag to a quorum, and this helps complete possibly failed writes as well.

We note that the server gossip is not necessary for correctness of CAS. We use `$\mathrm{gossip}$' in CAS mainly because it simplifies the proof of atomicity of the \emph{CASGC} algorithm, which is presented in Section \ref{sec:CASgc}. 

\subsection{Statements and proofs of correctness}
We next state the main result of this section.

\begin{mainth}\label{thm:dar}
CAS emulates shared atomic read/write memory.
\end{mainth}

To prove Theorem~\ref{thm:dar}, 
we show atomicity, Lemma~\ref{lem:atomicity}, and liveness, Lemma~\ref{lem:liveness}.
\subsubsection{Atomicity}

\begin{Nlemma}
\label{lem:atomicity}
CAS($k$) is atomic.
\end{Nlemma}

The main idea of our proof of atomicity involves defining, on the operations of any execution $\beta$ of CAS, a partial order $\prec$ that satisfies the sufficient conditions for atomicity described by Lemma 13.16 of \cite{Lynch1996}. We state these sufficient conditions in Lemma \ref{def:atomicity} next. 

\begin{Nlemma}[Paraphrased Lemma 13.16~\cite{Lynch1996}.]
\label{def:atomicity}
%
Suppose that the environment is well-behaved{,} meaning that an operation is invoked
at a client only if no other operation was performed by the client, or the client
received a response to the last operation it initiated.
Let $\beta$ be a (finite or infinite) execution of a read/write object, where 
$\beta$ consists of invocations and responses of read and write operations and 
where all operations terminate.  
Let $\Pi$ be the set of all operations in $\beta$.

Suppose that $\prec$ is an irreflexive partial ordering of all the operations in $\Pi$, 
satisfying the following properties:
{\bf (1)} If the response for $\pi_{1}$ precedes the invocation for $\pi_{2}$ in $\beta$,
	then it cannot be the case that $\pi_{2} \prec \pi_{1}$.
{\bf (2)} If $\pi_{1}$ is a write operation in $\Pi$ and $\pi_{2}$ is any operation in $\Pi$,
	then either $\pi_{1} \prec \pi_{2}$ or $\pi_{2} \prec \pi_{1}$.
{\bf (3)} The value returned by each read operation is the value written by the last 
 	preceding write operation according to $\prec$ (or $v_{0}$, if there is no such
	write).
\end{Nlemma}

{The following definition will be useful in defining a partial order on operations in an execution of CAS that satisfies the conditions of Lemma \ref{def:atomicity}.}
\begin{Ndefinition}
\label{def:tagofoperation}
Consider an execution $\beta$ of CAS and consider an operation $\pi$ that terminates in $\beta$. The \emph{tag} of operation $\pi,$ denoted as $T(\pi),$ is defined as follows: If $\pi$ is a read, then, $T(\pi)$ is the highest tag received in its \emph{query} phase. If $\pi$ is a write, then, $T(\pi)$ is the new tag formed in its \emph{pre-write} phase.
\end{Ndefinition}

We define our partial order $\prec$ as follows: In any execution $\beta$ of CAS, we order operations $\pi_1, \pi_2$ as $\pi_1 \prec \pi_2$ if
{\bf (i)} $T(\pi_1) < T(\pi_2),$ or 
{\bf (ii)} $T(\pi_1) = T(\pi_2),$ $\pi_1$ is a write and $\pi_2$ is a read.
We next argue that the partial ordering $\prec$ satisfies the conditions of \ref{def:atomicity}. We first show in Lemma \ref{lem:qfwproperty} that, in any execution $\beta$ of CAS, at any point after an operation $\pi$ terminates, the tag $T(\pi)$ has been propagated with the $\fin$ label to at least one quorum of servers. Intuitively speaking, Lemma \ref{lem:qfwproperty} means that if an operation $\pi$ terminates, the tag $T(\pi)$ is visible to any operation that is invoked after $\pi$ terminates. We crystallize this intuition in Lemma \ref{lem:tagsordering}, where we show that any operation that is invoked after an operation $\pi$ terminates acquires a tag that is at least as large as $T(\pi)$. Using Lemma \ref{lem:tagsordering} we show Lemma \ref{inv:tags}, which states that the tag acquired by each write operation is unique. Then we show that Lemma \ref{lem:tagsordering} and Lemma \ref{inv:tags} imply conditions {\bf (1)} and {\bf (2)} of Lemma \ref{def:atomicity}. By examination of the algorithm, we show that CAS also satisfies condition {\bf (3)} of Lemma \ref{def:atomicity}. 

\begin{Nlemma}
In any execution $\beta$ of CAS, for an operation $\pi$ that terminates in $\beta$, there exists a quorum $Q_{fw}(\pi)$ such that the following is true at every point of the execution $\beta$ after $\pi$ terminates:
 Every server of $Q_{fw}(\pi)$ has $(t, *, \fin)$ in its set of stored triples{, where} $*$ is 
either a coded element or `$\mathrm{null}$', and $t = T(\pi)$.
\label{lem:qfwproperty}
\end{Nlemma}
\begin{proof}
The proof 
is the same whether $\pi$ is a read or a write operation. The operation $\pi$ terminates after completing its \emph{finalize} phase, during which it receives responses from a quorum, say $Q_{fw}(\pi),$ to its \emph{finalize} message. This means that every server $s$ in $Q_{fw}(\pi)$ responded to the \emph{finalize} message from $\pi$ at some point before the point of termination of $\pi$. From the server protocol, we can observe that every server $s$ in $Q_{fw}(\pi)$ stores the triple $(t, *, \fin)$ at the point of responding to the \emph{finalize} message of $\pi$, where $*$ is either a coded element or $\Null$. Furthermore, the server $s$ stores the triple at every point after the point of responding to the \emph{finalize} message of $\pi$ and hence at every point after the point of termination of $\pi$.
\end{proof}
{\begin{Nlemma}
Consider any execution $\beta$ of CAS, and let $\pi_1,\pi_2$ be two operations that terminate in $\beta$. Suppose that $\pi_1$ returns before $\pi_2$ is invoked. Then $T(\pi_2) \geq T(\pi_1)$. Furthermore, if $\pi_2$ is a write, then $T(\pi_2) > T(\pi_1)$.
\label{lem:tagsordering}
\end{Nlemma}
\begin{proof}
To establish the lemma, it suffices to show that the tag acquired in the \emph{query} phase of $\pi_2,$ denoted as $\hat{T}(\pi_2),$ is at least as big as $T(\pi_1)$, that is, it suffices to show that $\hat{T}(\pi_2) \geq T(\pi_1)$. This is because, by examination of the client protocols, we can observe that if $\pi_2$ is a read, $T(\pi_2) = \hat{T}(\pi_2),$ and if $\pi_2$ is a write, $T(\pi_2) > \hat{T}(\pi_2)$. 

To show that $\hat{T}(\pi_2) \geq T(\pi_1)$ we use Lemma \ref{lem:qfwproperty}. We denote the quorum of servers that respond to the \emph{query} phase of $\pi_2$ as $\hat{Q}(\pi_2)$. We now argue that every server $s$ in $\hat{Q}(\pi_2) \cap Q_{fw}(\pi_1)$ responds to the \emph{query} phase of $\pi_2$ with a tag that is at least as large as $T(\pi_1)$. To see this, since $s$ is in $Q_{fw}(\pi_1)$, Lemma \ref{lem:qfwproperty} implies that $s$ has a tag $T(\pi_1)$ with label $\fin$ at the point of termination of $\pi_1$. Since $s$ is in $\hat{Q}(\pi),$ it also responds to the \emph{query} message of $\pi_2$, and this happens at some point after the termination of $\pi_1$ because $\pi_2$ is invoked after $\pi_1$ responds.
From the server protocol, we infer that server $s$ responds to the \emph{query} message of $\pi_2$ with a tag that is no smaller than $T(\pi_1)$. Because of Lemma \ref{lem:quorums}, there is at least one server $s$ in $\hat{Q}(\pi_2) \cap Q_{fw}(\pi_1)$ implying that operation $\pi_2$ receives at least one response in its \emph{query} phase with a tag that is no smaller than $T(\pi_1)$. Therefore $\hat{T}(\pi_2) \geq T(\pi_1)$.
\end{proof}

\begin{Nlemma}\label{inv:tags}
Let $\pi_1, \pi_2$ be write operations that terminate in an execution $\beta$ of CAS. Then $T(\pi_1) \neq T(\pi_2)$. 
\end{Nlemma}}
\begin{proof}
{Let $\pi_1,\pi_2$ be two write operations that terminate in execution $\beta.$ Let $C_1, C_2$ respectively indicate the identifiers of the client nodes at which operations $\pi_1, \pi_2$ are invoked. We consider two cases. \\
\emph{Case 1, $C_1 \neq C_2$:} From the write protocol, we note that $T(\pi_i) = (z_i, C_i)$. Since $C_1 \neq C_2$, we have $T(\pi_1) \neq T(\pi_2)$.\\
\emph{Case 2, $C_1 = C_2:$}} Recall that operations at the same client follow a ``handshake'' discipline, where a new invocation awaits the response of a preceding invocation. This means that one of the two operations $\pi_1, \pi_2$ should complete before the other starts. Suppose that, without loss of generality, the write operation $\pi_1$ completes before the write operation $\pi_2$ starts. Then, Lemma \ref{lem:tagsordering} implies that $T(\pi_2) > T(\pi_1)$. This implies that $T(\pi_2) \neq T(\pi_1)$.\qed
\end{proof}

\noindent \textit{Proof of Lemma \ref{lem:atomicity}.}
Recall that we define our ordering $\prec$ as follows: In any execution $\beta$ of CAS, we order operations $\pi_1, \pi_2$ as $\pi_1 \prec \pi_2$ if
{\bf (i)} $T(\pi_1) < T(\pi_2),$ or 
{\bf (ii)} $T(\pi_1) = T(\pi_2),$ $\pi_1$ is a write and $\pi_2$ is a read.

We first verify that the above ordering is {a partial order}, that is, if $\pi_1 \prec \pi_2,$ then it cannot be that $\pi_2 \prec \pi_1$. {We prove this by contradiction. Suppose that $\pi_1 \prec \pi_1$ and $\pi_2 \prec \pi_1$. Then, by definition of the ordering, we have that $T(\pi_1) \leq T(\pi_2)$ and vice-versa, implying that $T(\pi_1) = T(\pi_2)$. Since $\pi_1 \prec \pi_2$ and $T(\pi_1) = T(\pi_2)$, we have that $\pi_1$ is a write and $\pi_2$ is a read. But a symmetric argument implies that $\pi_2$ is a write and $\pi_1$ is a read, which is a contradiction. Therefore $\prec$ is a partial order.}

With the ordering $\prec$ defined as above, we now show {that the} three properties of Lemma \ref{def:atomicity} are satisfied. For property $\mathbf{(1)}$, consider an execution $\beta$ and two distinct operations $\pi_1, \pi_2$ in $\beta$ such that $\pi_1$ returns before $\pi_2$ is invoked. {If $\pi_2$ is a read, then Lemma \ref{lem:tagsordering} implies that $T(\pi_2) \geq T(\pi_1)$. By definition of the ordering, it cannot be the case that $\pi_2 \prec \pi_1$. If $\pi_1$ is a write, then Lemma \ref{lem:tagsordering} implies that $T(\pi_2) > T(\pi_1)$ and so, $\pi_1 \prec \pi_2$. Since $\prec$ is a partial order, it cannot be the case that $\pi_2 \prec \pi_1$.}

Property $\mathbf{(2)}$ follows from the definition of the $\prec$ in conjunction with Lemma \ref{inv:tags}.  

{Now we show property $\mathbf{(3)}$}: The value returned by each read operation is the value written by the last preceding write operation according to $\prec,$ or $v_0$ if there is no such write. Note that every version of the data object written in execution $\beta$ is \emph{uniquely} associated with a write operation in $\beta$. Lemma \ref{inv:tags} implies that every version of the data object being written can be uniquely associated with \emph{tag.} Therefore, to show that a read $\pi$ returns the last preceding write, we only need to argue that the read returns the value associated with $T(\pi)$. From the write, read, and server protocols, it is clear that a value and/or its coded elements are always paired together with the corresponding tags at every state of every component of the system. {In particular, the read returns the value from $k$ coded elements by inverting the MDS code $\Phi$; these $k$ coded elements were obtained at some previous point by applying $\Phi$ to the value associated with $T(\pi)$. Therefore Definition \ref{def:MDS} implies that the read returns the value associated with $T(\pi).$} \qed



\subsubsection{Liveness}
We now state the liveness condition satisfied by CAS. 

\remove{{It is more challenging to show the termination of a reader's \emph{finalize} phase.} Similar to the arguments for the \emph{query} phase, {we can show that} a quorum, say $Q_{fw}$ of servers responds to a reader's \emph{finalize} message. For the \emph{finalize} phase of a read to terminate, there is an additional requirement that at least $k$ servers include coded elements in their responses. To show this, suppose that the read acquired a tag $t$ in its \emph{query} phase. From examination of the CAS algorithm, we can infer that, at some point before the point of termination of the read's \emph{query} phase, a writer propagated a \emph{finalize} message with tag $t$. Let us denote by $Q_{pw}(t),$ the set of servers that responded to this write's \emph{pre-write} phase. Now, we argue that all servers in $Q_{pw}(t) \cap Q_{fw}$ respond to the reader's \emph{finalize} message with a coded element. To see this, let $s$ be any server in $Q_{pw}(t) \cap Q_{fw}.$
{
	Since $s$ is in $Q_{pw}(t)$, the server protocol for responding to a \emph{pre-write} message implies that $s$ has a coded element, $w_s$, at the point where it responds to that message.
	Since $s$ is in $Q_{fw},$ it also responds to the reader's \emph{finalize} message, and this happens at some point after it responds to the \emph{pre-write} message.
So it responds with its coded element $w_s.$}
 From Lemma \ref{lem:quorums}, it is clear that $|Q_{pw}(t) \cap Q_{fw}| \geq k$ implying that the reader receives at least $k$ coded elements in its \emph{finalize} phase and hence terminates.}}

\begin{Nlemma}[Liveness]
	\label{lem:liveness}
{CAS($k$) satisfies the following \emph{liveness} condition: {If $1 \leq k \leq N-2f$, {then every non-failing\footnote{An operation is said to have failed if the client performing the operation fails after its invocation but before its termination.} operation terminates} in every fair execution of CAS($k$)} where the number of server failures is no bigger than $f$ .}
\end{Nlemma}
\begin{proof}
{By examination }of the algorithm we observe that termination of any operation depends on 
termination of its phases. {So, to show liveness, we need to show that each phase of each operation terminates. Let us first examine the \emph{query} phase of a read/write operation; note that termination of the \emph{query} phase of a client is contingent on receiving responses from a quorum. Every non-failed server responds to a \emph{query} message with the highest locally available tag marked $\fin$. Since every server is initialized with $(t_0, v_0, \fin)$, every non-failed server has at least one tag associated with the label $\fin$ and hence responds to the client's \emph{query} message. Since the client receives responses from every non-failed server, property ({\bf ii}) of Lemma \ref{lem:quorums} ensures that the \emph{query} phase receives responses from at least one quorum, and hence terminates. We can similarly show that the \emph{pre-write} phase and \emph{finalize} phase of a writer terminate. In particular, termination of each of these phases is contingent on receiving responses from a quorum. Their termination {is }guaranteed from property ({\bf ii}) {of Lemma \ref{lem:quorums}} in conjunction with the fact that every non-failed server responds, at some point, to a \emph{pre-write} message and a \emph{finalize} message from a write with an acknowledgment. 

It remains to show the termination of a reader's \emph{finalize} phase. By using property ({\bf ii}) of Lemma \ref{lem:quorums}, we can show that a quorum, say $Q_{fw}$ of servers responds to a reader's \emph{finalize} message. For the \emph{finalize} phase of a read to terminate, there is an additional requirement that at least $k$ servers include coded elements in their responses. To show that this requirement is satisfied, suppose that the read acquired a tag $t$ in its \emph{query} phase. From examination of CAS, we infer that, at some {point before the point of termination of the read's \emph{query} phase}, a writer propagated a \emph{finalize} message with tag $t$. Let us denote by $Q_{pw}(t),$ the set of servers that responded to this write's \emph{pre-write} phase. We argue that all servers in $Q_{pw}(t) \cap Q_{fw}$ respond to the reader's \emph{finalize} message with a coded element. To see this, let $s$ be any server in $Q_{pw}(t) \cap Q_{fw}.$
{
	Since $s$ is in $Q_{pw}(t)$, the server protocol for responding to a \emph{pre-write} message implies that $s$ has a coded element, $w_s$, at the point where it responds to that message.
	Since $s$ is in $Q_{fw},$ it also responds to the reader's \emph{finalize} message, and this happens at some point after it responds to the \emph{pre-write} message.
So it responds with its coded element $w_s.$}
 From Lemma \ref{lem:quorums}, it is clear that $|Q_{pw}(t) \cap Q_{fw}| \geq k$ implying that the reader receives at least $k$ coded elements in its \emph{finalize} phase and hence terminates. }
\end{proof}
\subsection{Cost Analysis}
We analyze the communication costs of CAS in Theorem \ref{thm:comCAS}. The theorem implies that the 
read and write communication costs can be made as small as $\frac{N}{N-2f}$ \unit{} by choosing $k = N-2f.$

\begin{mainth}\label{thm:comCAS}
The write and read communication costs of the CAS($k$) are equal to $N/k$ \unit. 
\end{mainth}
\begin{proof}
For either protocol, observe that messages carry coded elements which have size $\frac{\log_{2}|\mathcal{V}|}{k}$ bits. More formally, each message is an element from $\mathcal{T}\times \mathcal{W} \times \{\pre, \fin\}$, 
where, $\mathcal{W}$ is a coded element corresponding to one of the $N$ outputs of the MDS code $\Phi$. As described in Sec. \ref{sec:ec}, $\log_{2}|\mathcal{W}| = \frac{\log_{2}|\mathcal{V}|}{k}.$
{The only messages that incur communication costs are the messages sent from the client to the servers in the \emph{pre-write} phase of a write and the messages sent from the servers to a client in the {\em finalize} phase of a read. It can be seen that the total communication cost of read and write operations of the CAS algorithm are $\frac{N}{k}\log_{2}|\mathcal{V}|$ bits, that is, they are upper bounded by this quantity and the said costs are incurred in certain worst-case executions.} 
\end{proof}

%

\section{Storage-Optimized Variant of CAS}\label{sec:CASgc}

Although CAS is efficient in terms of communication costs, it incurs an infinite storage cost because {servers} can store coded elements corresponding to an arbitrarily large number of versions. {We here present a variant of the CAS algorithm called \emph{CAS with Garbage Collection} (CASGC),  which has the same communication costs as CAS and incurs a bounded storage cost under {certain reasonable conditions.} CASGC achieves a bounded storage cost by using \emph{garbage collection}, i.e., by discarding coded elements with sufficiently small tags at the servers. CASGC is parametrized by two positive integers denoted as $k$ and $\delta$, where $1 \leq k \leq N-2f$; we denote the algorithm with parameter values $k, \delta$ by CASGC($k,\delta$). Like CAS($k$), we use an $(N,k)$ MDS code in CASGC($k,\delta$). The parameter $\delta$ is related to the number of coded elements stored at each server {under ``normal conditions'', that is, if all operations terminate and there are no ongoing write operations. }


\subsection{Algorithm description} {The CASGC$(k,\delta)$} algorithm is essentially the same as CAS$(k)$ with an additional garbage collection step at the servers. In particular,\str{ the client protocols in CASGC are identical to the corresponding protocols in CAS;} the only differences between the two algorithms {lie} in {the server actions} on receiving a \emph{finalize} message from a writer or a reader or `$\mathrm{gossip}$'. The server actions in the CASGC algorithm are described in Fig. \ref{fig:CASGC}. {In CASGC($k,\delta$), }each server stores the latest {$\delta+1$} triples with the $\fin$ label {plus} the triples corresponding to later and intervening operations with the $\pre$ label. {For the tags that are older (smaller) than the latest $\delta+1$ finalized tags received by the server, it stores only the metadata, not the data itself.} On receiving a {\em finalize} message either from a writer or a reader, {the server} performs a garbage collection step {before responding to the client}. The garbage collection step checks {whether} the server has more than {$\delta+1$} triples with the $\fin$ label. If so, it replaces the triple $(t', x , *)$ by $(t', \mbox{`$\mathrm{null}$'}, (*,\gc))$  for every tag $t'$ that is smaller than all the $\delta+1$ highest tags labeled $\fin$, where $*$ is $\pre$ or $\fin$, and $x$ can be a coded element or $\Null$. If a reader requests, through a \emph{finalize} message, a coded element that is already garbage collected, the server simply ignores this request. 
\begin{figure*}[t]
\begin{mdframed}
\footnotesize
\begin{tabbing}
{\bf servers }\\
        \ \ \begin{minipage}[t]{\textwidth}%
        \emph{state variable:} A variable that is a subset of $\mathcal{T} \times \left(\mathcal{W} \cup \{\Null\}\right) \times \{\pre, \fin,(\pre,\gc), (\fin, \gc)\}$ \\
            \emph{initial state}: Same as in Fig. \ref{fig:CAS}.  \\
            On receipt of {{\em query} message: Similar to Fig. \ref{fig:CAS}, respond with the highest locally available tag labeled $\fin$, i.e., respond with the highest $tag$ such that the triple $(tag, x, \fin)$ {or $(tag, \Null, (\fin, \gc))$ }is at the server, where $x$ can be a coded element or `$\mathrm{null}$'.}\smallskip 

            On receipt of a {\em pre-write} message: Perform the actions as described in Fig. \ref{fig:CAS} except the sending of an acknowledgement. Perform garbage collection. Then send an acknowledgement.  \smallskip \\
	    On receipt of a {\em finalize} from a writer: Let $t$ be the tag of the message. If a triple of the form $(t,x,\fin)$ or $(t,\Null,(\fin,\gc))$ is stored in the set of locally stored triples where $x$ can be a coded element or $\Null$, then ignore the incoming message. Otherwise, if a triple of the form $(t,w_s,\pre)$ or $(t,\Null,(\pre,\gc))$ is stored, then upgrade it to $(t,w_s,\fin)$ or $(t,\Null,(\fin,\gc))$. Otherwise, add a triple of the form $(t,\Null,\fin)$ to the set of locally stored triples. Perform garbage collection. Send `$\mathrm{gossip}$' message with item $(t,\fin)$ to all other servers. \smallskip \\
		On receipt of a {\em finalize} message from a reader: 
Let $t$ be the tag of the message. If a triple of the form $(t, w_s, *)$ exists in the list of stored triples where $*$ can be $\pre$ or $\fin$, then update it to $(t, w_s, \fin),$ perform garbage collection, and send $(t, w_s)$ to the reader. If $(t,\text{`}\mathrm{null}\text{'},(*,\text{`}\mathrm{gc}\text{'}))$ exists in the list of locally available triples where $*$ can be either $\fin$ or $\pre$, then update it to $(t,\text{`}\mathrm{null}\text{'},(\fin,\text{`}\mathrm{gc}\text{'}))$ and perform garbage collection, but do \emph{not} send a response.  Otherwise add $(t, \text{`}\mathrm{null}\text{'}, \fin)$ to the list of triples at the server, perform garbage collection, and send an {acknowledgment.}  Send `$\mathrm{gossip}$' message with item $(t,\fin)$ to all other servers.  \smallskip \\ 
On receipt of a `$\mathrm{gossip}$' message: Let $t$ denote the tag of the message. If a triple of the form $(t,x,\fin)$ or $(t,\Null,(\fin,\gc))$ is stored in the set of locally stored triples where $x$ can be a coded element or $\Null$, then ignore the incoming message. Otherwise, if a triple of the form $(t,w_s,\pre)$ or $(t,\Null,(\pre,\gc))$   is stored, then upgrade it to $(t,w_s,\fin)$ or $(t,\Null,(\fin,\gc))$. Otherwise, add a triple of the form $(t,\Null,\fin)$ to the set of locally stored triples. Perform garbage collection. \smallskip \\
{\em garbage collection:} If the total number of tags of the set $\{t: (t,x,*)\textrm{ is stored at the server, where } x \in \mathcal{W} \cup \{\Null\} \textrm{ and }* \in \{\fin, (\fin,\gc)\}\}$ is no bigger than $\delta+1,$ then return. Otherwise, let $t_1, t_2, \ldots t_{\delta+1}$ denote the highest $\delta+1$ tags from the set, sorted in descending order. Replace every element of the form $(t', x, *)$ where $t'$ is smaller than $t_{\delta+1}$ by $(t', \text{`}\mathrm{null}\text{'}, (*,\gc))$ where $*$ can be either $\pre$ or $\fin$ and $x \in \mathcal{W}\cup \{\Null\}$. \smallskip \\

\end{minipage} 
\end{tabbing}
\end{mdframed}
\caption{Server Actions for CASGC($k,\delta$).}\label{fig:CASGC}
\end{figure*}

\subsection{Statements and proofs of correctness} We next describe the correctness conditions satisfied by CASGC. We begin with a formal statement of atomicity. Later, we describe the liveness properties of CASGC.
\subsubsection{Atomicity}
\begin{mainth}[Atomicity]
CASGC is atomic. 
\label{thm:CASGCatomicity}
\end{mainth}
To show the above theorem, we observe that, from the perspective of the clients, the only difference between CAS and CASGC is in the server response to a read's \emph{finalize} message. In CASGC, when a coded element has been garbage collected, a server ignores a read's \emph{finalize} message. Atomicity follows similarly to CAS, since, in any execution of CASGC, operations acquire essentially the same tags as they would in an execution of CAS. We show this formally next. 
\begin{proof}[Proof (Sketch)]
Note that, formally, CAS is an I/O automaton formed by composing the automata of all the nodes and communication channels in the system. We show atomicity in two steps. In the first step, we construct a I/O automaton CAS$'$ which differs from CAS in that some of the actions of the servers in CAS$'$ are non-deterministic. However, we show that from the perspective of its external behavior (i.e., its invocations, responses and failure events), any execution of CAS$'$ can be extended to an execution of CAS implying that CAS$'$ satisfies atomicity. In the second step, we will show that CASGC {simulates} CAS$'$. These two steps suffice to show that CASGC satisfies atomicity.

We now describe CAS$'$. The CAS$'$ automaton is identical to CAS with respect to the client actions, and {to }the server actions on receipt of {\emph{query}} and {\emph{pre-write}} messages{ and {\emph{finalize}} messages from writers.} A server's response to a \emph{finalize} message from a read operation can be different in CAS$'$ as compared to CAS. In CAS$'$, at the point of the receipt of the \emph{finalize} message at the server, the server could respond either with the coded element, or {not respond at all }(even if it has the coded element).The server performs `$\mathrm{gossip}$' in CAS$'$ as in CAS.  

{We note that CAS$'$ ``simulates'' CAS. Formally speaking,} for every execution $\alpha{'}$ of CAS$'$, there is a natural corresponding execution $\alpha$ of CAS with an identical sequence of actions of all the components with one exception; when a server ignores a read's \emph{finalize} message in $\alpha{'}$, we assume that the corresponding message in $\alpha$ is indefinitely delayed. Therefore, from the perspective of client actions, for any execution $\alpha'$ of CAS$'$, there is an $\alpha$ of CAS with the same set of external actions. Since CAS satisfies atomicity, $\alpha$ has atomic behavior. Therefore $\alpha'$ is atomic, and implying that CAS$'$ satisfies atomicity.

{Now, we show that CASGC ``simulates'' CAS$'.$ That is, for every execution $\alpha_{\text{gc}}$ of CASGC, we construct a corresponding execution $\alpha'$ of CAS$'$ such that $\alpha'$ has the same external behavior (i.e., the same invocations, responses and failure events) as that of $\alpha_{\text{gc}}.$ We first describe the execution $\alpha'$ step-by-step, that is, we consider a step of $\alpha_{\text{gc}}$ and describe the corresponding step of $\alpha'$. We then show that the execution $\alpha'$ that we have constructed is consistent with the CAS$'$ automaton.

We construct $\alpha'$ as follows. We first set the initial states of all the components of $\alpha'$ {to be the same as they are in }$\alpha_{\text{gc}}.$  At every step, the states of the client nodes and the message passing system in $\alpha'$ are the same as the states of the corresponding components in the corresponding step of $\alpha_{\text{gc}}.$ A server's responses on receipt of a message is the same in $\alpha'$ as that of the corresponding server's response in $\alpha_{\text{gc}}$. In particular, we note that a server's external responses are the same in $\alpha_{\text{gc}}$ and $\alpha'$ even on receipt of a reader's \emph{finalize} message, that is, if a server ignores a reader's finalize message in $\alpha_{\text{gc}},$ it ignores the reader's finalize message in $\alpha'$ as well. Similarly, if a server sends a message as a part of `$\mathrm{gossip}$' in $\alpha_{\text{gc}}$, it sends a message in $\alpha'$ as well. The only difference between $\alpha_{\text{gc}}$ and $\alpha'$ is in the {change to the server's internal state} at a point of receipt of a \emph{finalize} message from a reader or a writer. At such a point, the server may perform garbage collection in $\alpha_{\text{gc}}$, whereas it does not perform garbage collection in $\alpha'$. Note that the initial state, the server's {response,} and the client states at every step of $\alpha'$ are the same as the corresponding step of $\alpha_{\text{gc}}.$ {Also note that 
 a server that fails at a step of $\alpha_{\text{gc}}$ fails at the corresponding step of $\alpha'$ (even though the server states could be different in general because of the garbage collection).} Hence, at every step, the external behavior of $\alpha'$ and $\alpha_{\text{gc}}$ are the same. This implies that the external behavior of the entire execution $\alpha'$ is the same as the external behavior of $\alpha_{\text{gc}}$. 

We complete the proof by noting that execution $\alpha'$ consistent with the CAS$'$ automaton. In particular, since the initial states of all the components are the same in the CAS$'$ and CASGC algorithms, the initial state of $\alpha'$ is consistent with the CAS$'$ automaton. Also, every step of $\alpha'$ is consistent with CAS$'$. Therefore, CASGC simulates CAS$'$. Since CAS$'$ is atomic, $\alpha_{\text{gc}}$ has atomic behavior. So CASGC is atomic.


}
\end{proof}

\subsubsection{Liveness}
Showing operation termination in CASGC is {more complicated} than CAS. 
This is because, in CASGC, when a reader requests a coded element, the server may have garbage collected it. The liveness property we show essentially articulates conditions under which read operations terminate in spite of the garbage collection. Informally speaking, we show that CASGC satisfies the following liveness property: every operation terminates in an execution where the number of failed servers is no bigger than $f$ and the number of writes \emph{concurrent} with a read is bounded by $\delta+1$. Before we proceed to formally state our liveness conditions, we give a formal definition of the notion of {concurrent operations in an execution} of CASGC. {For any operation} $\pi$ that completes its query phase, the tag of the operation $T(\pi)$ is defined as in Definition \ref{def:tagofoperation}. We begin with defining the \emph{end-point} of an operation. 
\begin{Ndefinition}[End-point of a write operation]
	In an execution $\beta$ of CASGC, the end point of a write operation $\pi$ in $\beta$ is defined to be 
	\begin{itemize*}
         \item[(a)] the first point of $\beta$ at which a quorum of servers that do not fail in $\beta$ has tag $T(\pi)$ with the $\fin$ label, where $T(\pi)$ is the tag of the operation $\pi$, if such a point exists,
		\item[(b)] the point of failure of operation $\pi,$ if operation $\pi$ fails and (a) is not satisfied.  
      	      \end{itemize*}
\end{Ndefinition}
{Note that if neither condition (a) nor (b) is satisfied, then the write operation has no end-point.
\begin{Ndefinition}[End-point of a read operation]
The end point of a read operation in $\beta$ is defined to be the point of termination if the read returns in $\beta$. The end-point of a failed read operation is defined to be the point of failure. 
\end{Ndefinition}
A read that does not fail or terminate has no end-point.}

\begin{Ndefinition}[Concurrent Operations]
	\label{def:concurrent}
			One operation is defined to be concurrent with another operation if it is not the case that the end point of either of the two operations is before the point of invocation of the other operation.
\end{Ndefinition}
Note that if both operations do not have end points, then they are concurrent with each other. 
We next describe the liveness property satisfied by CASGC.
\begin{mainth}[Liveness]
	Let $1 \leq k \leq N-2f$. Consider a fair execution $\beta$ of CASGC($k,\delta$) where the number of write operations concurrent to any read operation is at most $\delta,$ and the number of server node failures is at most $f$. Then, every {non-failing operation} terminates in $\beta$. \label{thm:CASGCliveness}
\end{mainth}

The main challenge in proving Theorem \ref{thm:CASGCliveness} lies in showing termination of read operations. In Lemma \ref{lem:invCASGC}, we show that if a \blue{read operation does not terminate} in an execution of CASGC($k,\delta)$, then the number of write operations that are concurrent with the read is larger than $\delta$. We then use the lemma to show Theorem \ref{thm:CASGCliveness} later in this section.
We begin by stating and proving Lemma \ref{lem:invCASGC}. 

\begin{Nlemma}
	Let $1 \leq k \leq N-2f$. Consider any fair execution $\beta$ of CASGC$(k,\delta)$ where the number of server failures is upper bounded by $f$. Let $\pi$ be a non-failing read operation in $\beta$ that does not terminate. Then, the number of writes that are concurrent with $\pi$ is at least $\delta+1$.
\label{lem:invCASGC}
\end{Nlemma}

To prove Lemma \ref{lem:invCASGC}, we prove Lemmas \ref{lem:CASGCpropagation} and \ref{lem:CASGC_operationorder}. Lemma \ref{lem:CASGCpropagation} implies that in a fair execution where the number of server failures is bounded by $f$, if a non-failing server receives a finalize message corresponding to a tag at some point, then the write operation corresponding to that tag has an end-point in the execution. We note that the server gossip plays a crucial role in showing Lemma \ref{lem:CASGCpropagation}. We then show Lemma \ref{lem:CASGC_operationorder} which states that in an execution, if a write operation $\pi$ has an end-point, then every operation that begins after the end-point of $\pi$ acquires a tag that is at least as large as the tag of $\pi$. Using Lemmas \ref{lem:CASGCpropagation} and \ref{lem:CASGC_operationorder}, we then show Lemma \ref{lem:invCASGC}. 

\begin{Nlemma}
	\blue{Let $1 \leq k \leq N-2f$. Consider any fair execution $\beta$ of CASGC$(k,\delta)$ where the number of server failures is no bigger than $f$. Consider a write operation $\pi$ that acquires tag $t$. If at some point of $\beta$, at least one non-failing server has a triple of the form $(t,x,\fin)$ or $(t,\Null,(\fin,\gc))$ where $x \in \mathcal{W}\cup \{\Null\}$, then operation $\pi$ has an end-point in $\beta$.}
\label{lem:CASGCpropagation}
\end{Nlemma}
\begin{proof}
Notice that every server that receives a \emph{finalize} message with tag $t$ invokes the `$\mathrm{gossip}$' protocol. If a \blue{non-failing} server $s$ stores tag $t$ with the $\fin$ label at some point of $\beta$, then from the server protocol we infer that it received a \emph{finalize} message with tag $t$ from a client or another server at some previous point. Since server $s$ receives the \emph{finalize} message with tag $t$, every \blue{non-failing} server also receives a \emph{finalize} message with tag $t$ at some point of the execution because of `$\mathrm{gossip}$'. Since a server that receives a \emph{finalize} message with tag $t$ stores the $\fin$ label after receiving the message, and the server does not delete the label associated with the tag at any point, eventually, every \blue{non-failing server} stores the $\fin$ label with the tag $t$. \blue{Since the number of server failures is no bigger than $f$, there is a quorum of non-failing servers that stores tag $t$ with the $\fin$ label at some point of $\beta$. Therefore, operation $\pi$ has an end-point in $\beta$, with the end-point being the first point of $\beta$ where a quorum of non-failing servers have the tag $t$ with the $\fin$ label.}
\end{proof}

\begin{Nlemma}
\blue{ Consider any execution $\beta$ of CASGC$(k,\delta).$ If write operation $\pi$ with tag $t$ has an end-point in $\beta$, then the tag of any operation that begins after the end point of $\pi$ is at least as large as $t$.}
	\label{lem:CASGC_operationorder}
\end{Nlemma}
\begin{proof}
Consider a write operation $\pi$ that has an end-point in $\beta$. \blue{By definition, at the end-point of $\pi,$ there exists at least one quorum $Q(\pi)$ of non-failing servers such that each server has the tag $t$ with the $\fin$ label. Furthermore, from the server protocol, we infer that each server in quorum $Q(\pi)$ has the tag $t$ with the $\fin$ label at every point after the end point of the operation $\pi$. }

Now, suppose operation $\pi'$ is invoked after the end point of $\pi$. We show that the tag acquired by operation $\pi'$ is at least as large as $t$. Denote the quorum of servers that respond to the \emph{query} phase of $\pi'$ as ${Q}(\pi')$. We now argue that every server $s$ in ${Q}(\pi) \cap Q(\pi')$ responds to the \emph{query} phase of $\pi'$ with a tag that is at least as large as $t$. To see this, since $s$ is in $Q(\pi)$, it has a tag $t$ with label $\fin$ at the end-point of $\pi$. Since $s$ is in ${Q}(\pi'),$ it also responds to the \emph{query} message of $\pi'$, and this happens at some point after the end-point of $\pi$ because $\pi'$ is invoked after the end-point of $\pi$. Therefore server $s$ responds with a tag that is at least as large as $t$. This completes the proof.
\end{proof}

\begin{proof}[Proof of Lemma \ref{lem:invCASGC}]
	Note that the termination of the query phase of the read is contingent on receiving a quorum of responses. By noting that every non-failing server responds to the read's query message, we infer from Lemma \ref{lem:quorums} that the query phase terminates.  \blue{It remains to consider termination of the read's finalize phase.}
	Consider an operation $\pi$ whose finalize phase does not terminate. We argue that there are \blue{at least} $\delta+1$ write operations that are concurrent with $\pi$. 

Let $t$ be the tag acquired by operation $\pi$. By property ({\bf ii}) of Lemma \ref{lem:quorums}, we infer that a quorum, say $Q_{fw}$ of \blue{non-failing} servers receives the read's \emph{finalize} message. There are only two possibilities. 

$\mathbf{(i)}$ There is no server $s$ in $Q_{fw}$ such that, at the point of receipt of the read's finalize message at server $s$, a triple of the form $(t,\Null, (*,\gc))$ exists at the server.

$\mathbf{(ii)}$ There is at least one server $s$ in $Q_{fw}$ such that, at the point of receipt of the read's finalize message at server $s$, a triple of the form $(t,\Null, (*,\gc))$ exists at the server. 

In case $\mathbf{(i)},$ we argue in a manner that is similar to Lemma \ref{lem:liveness} that the read receives responses to its finalize message from quorum $Q_{fw}$ of which at least $k$ responses include coded elements. \blue{We repeat the argument here for completeness. From examination of CASGC, we infer that, at some {point before the point of termination of the read's \emph{query} phase}, a writer propagated a \emph{finalize} message with tag $t$. Let us denote by $Q_{pw}(t),$ the set of servers that responded to this write's \emph{pre-write} phase. We argue that all servers in $Q_{pw}(t) \cap Q_{fw}$ respond to the reader's \emph{finalize} message with a coded element. To see this, let $s'$ be any server in $Q_{pw}(t) \cap Q_{fw}.$ Since $s'$ is in $Q_{pw}(t)$, the server protocol for responding to a \emph{pre-write} message implies that $s'$ has a coded element, $w_{s'}$, at the point where it responds to that message. Since $s'$ is in $Q_{fw},$ it does not contain an element of the form $(t,\Null,(*,\gc))$ implying that it has not garbage collected the coded element at the point of receipt of the reader's finalize message. Therefore, it responds to the reader's \emph{finalize} message, and this happens at some point after it responds to the \emph{pre-write} message. So it responds with its coded element $w_{s'}$.  From Lemma \ref{lem:quorums}, it is clear that $|Q_{pw}(t) \cap Q_{fw}| \geq k$ implying that the reader receives at least $k$ coded elements in its \emph{finalize} phase and hence terminates. Therefore the finalize phase of $\pi$ terminates, {contradicting our assumption that it does not. Therefore $\mathbf{(i)}$ is impossible.}}

We next argue that in case $\mathbf{(ii)},$ there are at least $\delta+1$ write operations that are concurrent with the read operation $\pi$. In case $\mathbf{(ii)}$, from the server protocol of CASGC, we infer that at the point of receipt of the reader's finalize message at server $s$, there exist tags $t_1, t_2,\ldots,t_{\delta+1},$ each bigger than $t$, such that a triple of the form $(t_i,x,\fin)$ or $(t_i, \Null, (\fin,\gc))$ exists at the server. We infer from the write and server protocols  that, for every $i$ in $\{1,2,\ldots,\delta+1\},$ a write operation, say $\pi_i,$ must have committed to tag $t_i$ in its \emph{pre-write} phase before this point in $\beta$. \blue{Because $s$ is non-failing in $\beta$, we infer from Lemma \ref{lem:CASGCpropagation} that operation $\pi_i$ has an end-point in $\beta$ for every $i \in \{1,2,\ldots,\delta+1\}$. Since $t < t_i$ for every $i \in \{1,2,\ldots,\delta+1\}$, we infer from Lemma \ref{lem:CASGC_operationorder} that the end point of write operation $\pi_i$ is after the point of invocation of operation $\pi$.} Therefore operations $\pi_1,\pi_2,\ldots,\pi_{\delta+1}$ are concurrent with read operation $\pi$.

 \end{proof}

\blue{A proof of Theorem \ref{thm:CASGCliveness} follows from Lemma \ref{lem:invCASGC} in a manner that is similar to Lemma \ref{lem:liveness}. We briefly sketch the argument here.}

\begin{proof}[Proof Sketch of Theorem \ref{thm:CASGCliveness}]
\blue{
{By examination }of the algorithm we observe that termination of any operation depends on 
termination of its phases. So, to show liveness, we need to show that each phase of each operation terminates. We first consider a write operation. Note that termination of the \emph{query} phase of a write operation is contingent on receiving responses from a quorum. Every non-failed server responds to a \emph{query} message with the highest locally available tag marked $\fin$. Since every server is initialized with $(t_0, v_0, \fin)$, every non-failed server has at least one tag associated with the label $\fin$ and hence responds to the writer's \emph{query} message. Since the writer receives responses from every non-failed server, property ({\bf ii}) of Lemma \ref{lem:quorums} ensures that the \emph{query} phase receives responses from at least one quorum, and hence terminates. We can similarly show that the \emph{pre-write} phase and \emph{finalize} phase of a writer terminate. }

\blue{It remains to consider the termination of a read operation. Suppose that a non-failing read operation does not terminate. Then, from Lemma \ref{lem:invCASGC}, we infer that there are at least $\delta+1$ writes that are concurrent with the read. This contradicts our assumption that the number of write operations that are concurrent with a read is no bigger than $\delta$. Therefore every non-failing read operation terminates.}
\end{proof}

\subsection{Bound on storage cost} 

We bound the storage cost of an execution of CASGC by providing a bound on the number of coded elements stored at a server {at \blue{any particular point} of the execution}. In particular, in Lemma \ref{lem:CASGC_storage}, we describe conditions under which coded elements corresponding to the value of a write operation {are garbage collected} at \emph{all} the servers. Lemma \ref{lem:CASGC_storage} naturally leads to \blue{a} storage cost bound in Theorem \ref{thm:storCAS}. {We begin with a definition of an \emph{$\omega$-\blue{superseded} write operation} \blue{for a point in an execution}, for a positive integer $\omega$.}
\begin{Ndefinition}[$\omega$-superseded write operation]
	In an execution $\beta$ of CASGC, consider a write operation $\pi$ that completes its query phase. Let $T(\pi)$ denote the tag of the write. Then, the write operation is said to be $\omega$-superseded at a point $P$ of the execution if \blue{there are at least $\omega$ terminating write operations, each with a tag that is bigger than $T(\pi),$ such that every message on behalf of each of these operations (including `$\mathrm{gossip}$' messages) has been delivered by point $P$. }
\end{Ndefinition}
We show} in Lemma \ref{lem:CASGC_storage} that in {an execution of CASGC$(k,\delta)$}, if a write operation is \blue{$(\delta+1)$-superseded} at a point, then, no server stores a coded element corresponding to the operation at that point because of garbage collection. We state and prove Lemma \ref{lem:CASGC_storage} next. We then use Lemma \ref{lem:CASGC_storage} to describe a bound on the storage cost of any execution of CASGC($k,\delta$) in Theorem \ref{thm:storCAS}.

\begin{Nlemma}
	Consider an execution $\beta$ of CASGC($k,\delta$) and consider any point $P$ of $\beta$. {If} a write operation $\pi$ is \blue{$(\delta+1)$-superseded} at point $P$, {then no non-failed server} has a coded element corresponding to the value of the write operation $\pi$ at point $P$. \label{lem:CASGC_storage}
\end{Nlemma}
\begin{proof}[Proof]
Consider an execution $\beta$ of CASGC($k,\delta$) and a point $P$ in $\beta$. Consider a {write} operation $\pi$ that is \blue{$(\delta+1)$-superseded} at point $P$. Consider an arbitrary server $s$ that has not failed at point $P$. We show that server $s$ does not have a coded element corresponding to operation $\pi$ at point $P.$ Since operation $\pi$ is \blue{$(\delta+1)$-superseded} at point $P$, there exist at least $\delta+1$ write operations $\pi_1, \pi_2, \ldots, \pi_{\delta+1}$ such that, for every $i \in \{1,2,\ldots,\delta+1\}$,
\begin{itemize*}
\item operation $\pi_i$ terminates in $\beta,$ 
\item the tag $T(\pi_i)$ acquired by operation $\pi_i$ is larger than $T(\pi)$, and 
\item every message on behalf of operation $\pi_i$ is delivered by point $P$.
\end{itemize*} 
Since operation $\pi_i$ terminates, it completes its \emph{finalize} phase where it sends a finalize message with tag $T(\pi_i)$ to server $s$. Furthermore, the \emph{finalize} message with tag $T(\pi_i)$ arrives at server $s$ by point $P$. Therefore, by point $P$, server $s$ has received at least $\delta+1$ finalize messages, one from each operation in $\{\pi_i:i=1,2,\ldots,\delta+1\}$. The garbage collection executed by the server on the receipt of the last of these finalize messages ensures that the coded element corresponding to tag $T(\pi)$ does not exist at server $s$ at point $P$. This completes the proof.

\end{proof}

\begin{mainth}
\label{thm:storCAS}
Consider an execution $\beta$ of CASGC($k,\delta$) such that, at any point of the execution, \blue{the number of writes that have completed their query phase by that point and are not $(\delta+1)$-superseded at that point is upper bounded by $w$.} The storage cost of the execution is at most $\frac{w N}{k} \log_2|\mathcal{V}|.$
\end{mainth}
\begin{proof}

		\blue{Consider an execution $\beta$ where at any point of the execution, the number of writes that have completed their query phase by that point and are not $(\delta+1)$-superseded at that point is upper bounded by $w$. Consider an arbitrary point $P$ of the execution $\beta$, and consider a server $s$ that is non-failed at point $P$. We infer from the write and server protocols that, at point $P$, server $s$ does not store a coded element corresponding to any write operation that has not completed its query phase by point $P$. We also infer from Lemma \ref{lem:CASGC_storage} that server $s$ does not store a coded element corresponding to an operation that is $(\delta+1)$-superseded at point $P$. Therefore, if server $s$ stores a coded element corresponding to a write operation at point $P$, we infer that the write operation has completed its query phase but is not $(\delta+1)$-superseded by point $P$. By assumption on the execution $\beta$, the number of coded elements at point P of $\beta$ at server $s$ is upper bounded by $w$. Since each coded element has a size of $\frac{1}{k}$ \units and we considered an arbitrary server $s$, the storage cost at point $P,$ summed over all the non-failed servers, is upper bounded by $\frac{wN}{k}$ \unit. Since we considered an arbitrary point $P$, the storage cost of the execution is upper bounded by $\frac{wN}{k}$ \unit.}
	\end{proof}
	We note that Theorem \ref{thm:storCAS} can be used to obtain a bound \blue{on} the storage cost of executions in terms of various parameters of the system components. For instance, the theorem can be used to obtain a bound on the storage cost in terms of an upper bound on the delay of every message, the number of steps for the nodes to take actions, the rate of write operations, and the rate of failure.
In particular, the above parameters can be used to bound the number of writes that are not $(\delta+1)$-superseded, which can then be used to bound the storage cost.

{
\section{Communication Cost Optimal Algorithm}
\label{sec:lowerbound}
A natural question is whether one might be able to prove a lower bound to show that communication costs of CAS and CASGC are optimal. Here, we describe a new \blue{``counterexample algorithm''} called \emph{Communication Cost Optimal Atomic Storage} (CCOAS) algorithm, which shows that such a lower bound cannot be proved. We show in Theorem \ref{thm:storCCOAS} that CCOAS has write and read communication costs of $\frac{N}{N-f}\log_2|\mathcal{V}|$ bits, which is smaller than the communication costs of CAS and CASGC. Because elementary coding theoretic bounds imply that these costs can be no smaller than $\frac{N}{N-f}\log_2|\mathcal{V}|$ bits, CCOAS is optimal from the perspective of communication costs. CCOAS, however, is infeasible in practice because of certain drawbacks described later in this section.
\begin{figure}[tb]
\begin{mdframed}
\footnotesize
\begin{tabbing}
{\bf write}(\emph{value})\\
	\ \ \begin{minipage}[t]{\textwidth}%
\emph{query}: Same as in CAS($N-2f$). \\
\emph{pre-write}: Select the largest tag from the \emph{query} phase; form a new tag $t$ by incrementing integer by 1 and adding its `id'. {{Apply an $(N,N-f)$ MDS code $\Phi$ to \emph{value}}} and obtain coded elements $w_1,\dots,w_N$. Send $(t, w_s,\pre)$ to every server $s$. Await responses from a quorum.

\emph{finalize}: Same as in CAS($N-2f$). 
\end{minipage}\\ \\
{\bf read} \\
	\ \ \begin{minipage}[t]{\textwidth}%

\emph{query}: Same as in CAS($N-2f$).

\emph{finalize}: Select largest tag $t$ from the query phase. Send \emph{finalize} message $(t,\Null,\fin)$ to all servers requesting the associated coded elements. { Await responses with coded elements from a quorum.} Obtain the \emph{value} by inverting $\Phi$, and terminate by returning $value$.
\end{minipage}
\\ \\
{\bf server}\\
	\ \ \begin{minipage}[t]{\textwidth}%

\emph{state variables}: State is a subset of $\mathcal{T} \times (\mathcal{W}\cup \{\Null\}) \times \{ \pre, \fin\} \times 2^{\mathcal{C}}$.

\emph{initial state}: $(t_0,w_{0,s},\fin, \{\})$.

Response to \emph{query}: Send highest locally known tag that has label $\fin$.

Response to \emph{pre-write}: If the tag $t$ of the message is not available in the locally stored set of tuples, add the tuple $(t, w_s, \pre, \{\})$ to the locally stored set. {If $(t, \Null, \fin, \mathcal{C}_{0})$ exists in the locally stored set of tuple for some set of clients $\mathcal{C}_{0}$, then send $(t, w_s)$ to every client in $\mathcal{C}_{0}$ and modify the locally stored tuple to $(t, w_s, \fin, \{\}).$} Send acknowledgement to the writer.

Response to \emph{finalize} of write: Let $t$ denote the tag of the message. If $(t,w_s,\pre, \{\})$ exists in the locally stored set of tuple where $*$ can be $\pre$ or $\fin$, update to $(t,w_s,\fin, \{\})$. If no tuple exists in the locally stored set with tag $t$, add $(t,\Null,\fin, \{\})$ to the locally stored set. Send acknowledgement.

Response to \emph{finalize} of read: Let $t$ denote the tag of the message and $C \in \mathcal{C}$ denote the identifier of the client sending the message. If $(t,w_s,*, \mathcal{C}_{0})$ exists in the locally stored set, update the tuple as  $(t,w_s,\fin, \mathcal{C}_0)$ and send $(t,w_s)$ to reader. If $(t, \Null, \fin, \mathcal{C}_{0})$ exists at the server, update it as $(t, \Null, \fin, \mathcal{C}_{0}\cup \{C\})$. Otherwise, add  $(t, \Null, \fin, \{C\})$ to the list of locally stored tags. 
\end{minipage}
\end{tabbing}
\end{mdframed}
\caption{The CCOAS algorithm. We denote the (possibly infinite) set of clients by $\mathcal{C}.$ The notation $2^\mathcal{C}$ denotes the power set of the set of clients $\mathcal{C}$.}
\label{fig:counterexample}
\end{figure}

\subsection{Algorithm description}
CCOAS resembles CAS in its structure. Like CAS($N-2f$), its quorum $\mathcal{Q}$ consists of the set of all subsets of $\mathcal{N}$ that have at least $N-f$ elements. We also use terms ``query'', ``pre-write'', and ``finalize'' for the various phases of operations.  We provide a formal description of CCOAS in Fig. \ref{fig:counterexample}. Here, we informally describe the differences between CAS and CCOAS.
\begin{itemize*}
\item In CCOAS, the writer uses an $(N,N-f)$ MDS code to generate coded elements. Note the contrast with CAS($k$) which uses an $(N,k)$ code, where the parameter $k$ is at most $N-2f.$ Because we use an $(N,N-f)$ code in CCOAS, the size of each coded element is equal to $\frac{\log_2|\mathcal{V}|}{N-f}$ bits, and as a consequence, the read and write communication costs are equal to $\frac{N}{N-f}\log_2|\mathcal{V}|$ bits.
\item In CCOAS, a reader requires $N-f$ responses with coded elements for termination of its finalize phase. In CAS, in general, at most $N-2f$ responses with coded elements are required. 
\item In CCOAS, the servers respond to finalize messages from a read with coded elements only. This is unlike CAS, where a server that does not have a coded element corresponding to the tag of a reader's finalize message at the point of reception responds simply with an acknowledgement. In CCOAS, if a server does not have a coded element corresponding to the tag $t$ of a reader's finalize message at the point of reception, then, in addition to adding a triple of the form $(t, \Null, \fin)$ to its local storage, the server registers this read along with tag $t$ in its logs. When the corresponding coded element with tag $t$ arrives at a later point, the server, in addition to storing the coded element, sends it to every reader that is registered with tag $t$. We show in our proofs of correctness that, in CCOAS, every non-failing server responds to a finalize message from a read with a coded element at some point. 
\end{itemize*}

\subsection{Proof of correctness and communication cost}
We next describe a formal proof of the correctness of CCOAS. 
\subsubsection{Atomicity}
\begin{mainth}
	CCOAS emulates shared atomic read/write memory. 
	\label{thm:counterexample}
\end{mainth}

The main challenge in showing Theorem \ref{thm:counterexample} lies in showing termination of read operations, specifically to show that every non-failing server sends a coded element in response to a reader's finalize message. The theorem follows from Lemmas \ref{lem:counterexampleliveness} and \ref{lem:counterexampleatom}, which are stated next.

\begin{Nlemma}
The CCOAS algorithm satisfies atomicity.
\label{lem:counterexampleatom}
\end{Nlemma}
\begin{proof}
Atomicity can be shown via a simulation relation with CAS.  We provide a brief informal sketch of the relation here. We argue that for every execution $\beta$ of CCOAS, there is an execution $\beta'$ of CAS with the same trace. To see this, we note that the write protocol of CCOAS is essentially identical to the write protocol in CAS, with the only difference between the two algorithms being the erasure code used in the pre-write phase. Similarly, the query phase of the read protocols of both algorithms are the same. Also note that the server responses to messages from a writer and query messages from a reader are identical in both CAS and CCOAS. The main differences between CCOAS and CAS in the server actions. The first difference is that, in CCOAS, the servers do not perform `$\mathrm{gossip}$'. The second difference is that in CCOAS, if the server does not have a coded element corresponding to the tag of the reader's finalize message, then the server does not respond at this point. Instead, the server sends a coded element to the reader at the point of receipt of the pre-write message with this tag. We essentially create $\beta'$ from $\beta$ by delaying all messages `$\mathrm{gossip}$' messages indefinitely, and delaying reader's finalize messages so that they arrive at each server at the point of, or after the receipt of the corresponding pre-write message by the server. This delaying ensures that the server actions are identical in both $\beta$ and $\beta'$.

Specifically, we create $\beta'$ as follows. In $\beta'$ the points of 
\begin{itemize*}
\item invocations of operations, 
\item sending and receipt of messages between writers and servers, 
\item sending and receipt of query messages between readers and servers, 
\item and sending of finalize messages from the readers
\end{itemize*}
are identical to $\beta$. The server `$\mathrm{gossip}$' messages in $\beta'$ are delayed indefinitely. A crucial difference between $\beta$ and $\beta'$ lies in the points of receipt of reader's finalize messages at the servers. Consider a read operation that acquired tag $t$ in $\beta$ and let $P$ denote the point of receipt of a reader's finalize message to server $s$. Let $P'$ denote the point of receipt of a pre-write message with tag $t$ at server $s$ in $\beta$.  Now, consider the corresponding read operation that acquired tag $t$ in $\beta'.$ Now, if $P$ precedes $P'$ in $\beta,$ then the reader's finalize message with tag $t$ arrives at server $s$ at $P'$ in $\beta'$, else, it arrives at point $P$ in $\beta'$. This implies that server $s$ responds to reader's finalize messages at the same points in $\beta$ and $\beta'$. Finally, we complete our specification of $\beta'$ by letting a server's response to the reader's finalize message arrive at the client at the same point in $\beta'$ as in $\beta$.

Note that if an operation acquires tag $t$ in $\beta$, the corresponding operation in $\beta'$ also acquires tag $t$. Also note that the points of invocation, responses of operations and the values returned by read operations are the same in both $\beta$ and $\beta'$. Therefore, there exists an execution $\beta'$ of CAS with the same trace as an arbitrary execution $\beta$ of CCOAS. Since CAS is atomic, $\beta'$ has atomic behavior, and so does $\beta$. Therefore, CCOAS satisfies atomicity.
\end{proof}

\subsubsection{Liveness}
We next state the liveness condition of CCOAS.
\begin{Nlemma}
\label{lem:counterexampleliveness}
CCOAS satisfies the liveness condition: in every fair execution where the number of failed servers is no bigger than $f$, every non-failing operation terminates.
\end{Nlemma}
To show Lemma \ref{lem:counterexampleliveness}, we first state and prove Lemma \ref{lem:counterexample}. Informally speaking, Lemma \ref{lem:counterexample} implies that every non-failing server responds to a reader's finalize message with a coded element. As a consequence, every read operation gets $N-f$ coded elements in response to its finalize messages. Therefore its finalize phase implying that the operation returns implying Lemma \ref{lem:counterexampleliveness}. We first state and prove Lemma \ref{lem:counterexample}. Then we prove Lemma \ref{lem:counterexampleliveness}. 

\begin{Nlemma}
Consider any fair execution $\alpha$ of CCOAS and a server $s$ that does not fail in $\alpha$. Then, for any read operation in $\alpha$ with tag $t$, the server $s$ responds to the read's finalize message with the coded element corresponding to tag $t$ at some point of $\alpha$. 
\label{lem:counterexample}
\end{Nlemma}

\begin{proof}[Proof sketch]
	Consider a server $s$ that does not fail in $\alpha$ and consider the point $P$ of $\alpha$ where server $s$ receives a finalize message with tag $t$ from a reader. Since the read operation at the reader acquired tag $t$, from examination of the algorithm we can infer that a write with tag $t$ completed its pre-write phase at some point of $\alpha$. From the write protocol, note that this implies that the writer sent a coded element with tag $t$ to every server in its pre-write phase. In particular, the writer sent coded element $w_s$ to server $s$. Since the channels are reliable and since $s$ does not fail in $\alpha$, this means that at some point $P'$ of $\alpha$, the server $s$ receives the coded element $w_s$.  There are only two possible scenarios. First, $P'$ precedes $P$ in $\alpha$, and second, $P$ precedes $P'$. To complete the proof, we show that, in the first scenario the server responds to the reader's finalize message with $w_s$ at point $P$, and in the second scenario\footnote{Note that in this second scenario, the server does not respond with a coded element in CAS, where the server only sends an acknowledgement. In contrast to the proof here, the liveness proof of CAS involved showing that at least $k$ servers satisfy the condition imposed by the first scenario.}, the server responds to the reader's finalize message with $w_s$ at point $P'$. 

In the first scenario, note that the server has a coded element $w_s$ at the point $P$. By examining the server protocol, we observe that server $s$ responds to the reader's finalize message with a coded element $w_s$. 

In the second second scenario, point $P'$ comes after $P$ in $\alpha$. Because of the server protocol on receipt of the reader's finalize message, server $s$ adds a tuple of the form $(t, \Null, \fin, \mathcal{C}_{0}),$ where $C \in \mathcal{C}_{0},$ to the local state at point $P$. Also, note that, at point $P'$, the server stores a tuple of the form $(t, \Null, \fin, \mathcal{C}_{1}),$ where $C \in \mathcal{C}_1.$  Finally, based on the server protocol on receipt of a pre-write message, we note that at point $P'$,  the server sends $w_s$ to all the clients in $\mathcal{C}_{1}$ including client $C$. This completes the proof.
\end{proof}

We next prove Lemma \ref{lem:counterexampleliveness}.

\begin{proof}[Proof of Lemma \ref{lem:counterexampleliveness}]
To prove liveness, it suffices to show that in any fair execution $\alpha$ where at most $f$ servers fail, every phase of every operation terminates. The proof of termination of a write operation, and the query phase of a read operation is similar to CAS and omitted here for brevity. Here, we present a proof of termination of the finalize phase of a read in any fair execution $\alpha$ where at most $f$ servers fail.

To show the termination of a read, note from Lemma \ref{lem:counterexample} that in execution $\alpha$, every non-failed server $s$ responds to a reader's finalize message with a coded element. Because the number of servers that fail in $\alpha$ is at most $f$, this implies that reader obtains at least $N-f$ messages with coded elements in response to its finalize message. From the read protocol, we observe that this suffices for termination of the finalize phase of a read. This completes the proof.
\end{proof}

\subsubsection{Communication cost}
We next state the communication cost of CCOAS.
\begin{mainth}
	The write and read communication costs of CCOAS are both equal to $\frac{N}{N-f}\log|\mathcal{V}|$.
\label{thm:storCCOAS}
\end{mainth}
The proof of Theorem \ref{thm:storCCOAS} is similar to the proof of Theorem \ref{thm:comCAS} and is omitted here for brevity.

\subsection{Drawbacks of CCOAS} CCOAS incurs a smaller communication cost mainly because the reader acquires $N-f$ coded elements, thus allowing the writer to use an $(N,N-f)$ MDS code. Since a write operation returns after getting responses from some quorum, there are executions of our algorithm where, at the point of termination of a write operation, only a quorum $Q_{pw}$ containing $N-f$ servers have received its pre-write messages. Now, if one of the servers in $Q_{pw}$ fails after the termination of the write, then, since a reader that intends to acquire the value written requires $N-f$ coded elements, it is important that at least one of the pre-write messages sent by the writer to a server outside of $Q_{pw}$ reaches the server. In other words, it is crucial for liveness of read operations that the pre-write messages sent by the write operation are delivered to every non-failing server, even if some of these messages have not been delivered at the point of termination of the write. We use this assumption implicitly in the proof of correctness of CCOAS. 


Although, in our model, channels deliver messages of operations that have terminated, the dependence of liveness on this assumption is a significant drawback of CCOAS. The modeling assumption of reliable channels is often an implicit abstraction of a lossy channel and an underlying primitive that retransmits lost messages until they are delivered. From a practical point of view, however, it is not well-motivated to assume that this underlying primitive retransmits lost messages corresponding to operations that have terminated, especially if the client performing the operation fails. \blue{We note that CAS and CASGC do not share this drawback of CCOAS.} An interesting future exercise is to generalize CAS and CASGC to lossy channel models (see, for example, the model used in \cite{DGL}).
}

\section{Conclusions}\label{sec:fin}
We have proposed low-cost algorithms for atomic shared memory emulation in 
asynchronous message-passing systems. We also contribute to this body of work through rigorous definitions and analysis of (worst-case) communication and storage costs. We show that our algorithms have desirable properties in terms of the amount of communication and storage costs.
{There are several relevant follow up research directions in this topic. An interesting question is whether the storage cost can be reduced through a more sophisticated coding strategy, for instance, using the code constructions of \cite{Wang_Cadambe_MultiversionCoding}. We note that when erasure coding is used for shared memory emulation, the communication and storage costs of various algorithms seem to depend on the number of parallel operations in the system. For instance, in all the erasure coding-based algorithms, servers store coded elements corresponding to multiple versions at the servers. Similarly, in ORCAS-B and HGR, servers send coded elements corresponding to multiple versions to the reader. A natural question is whether there exist fundamental lower bounds that capture this behavior, or whether there exist algorithms that can achieve low communication and storage costs which \blue{do not grow} with the extent of parallelism in the system.} Among the remaining questions, we emphasize the need for generalizing of CAS and CASGC to lossy channels, and to dynamic settings possibly through modifications of RAMBO \cite{RAMBO}.

\newpage
\bibliographystyle{abbrv}
\bibliography{biblio}
\newpage
\appendix
\section{Descriptions of the ABD and LDR Algorithms}
\label{app:ABD_LDR}

As baselines for our work we use the MWMR versions of {the }ABD and LDR algorithms \cite{ABD,LDR}. Here, we describe the ABD and LDR algorithms, and evaluate their communication and storage costs. We present the ABD and LDR algorithms in Fig. \ref{fig:ABD} and Fig. \ref{fig:LDR} respectively.
 The costs of these algorithms are stated 
in Theorems \ref{claim:ABD} and \ref{claim:comLDR}. 

\begin{figure}[h]
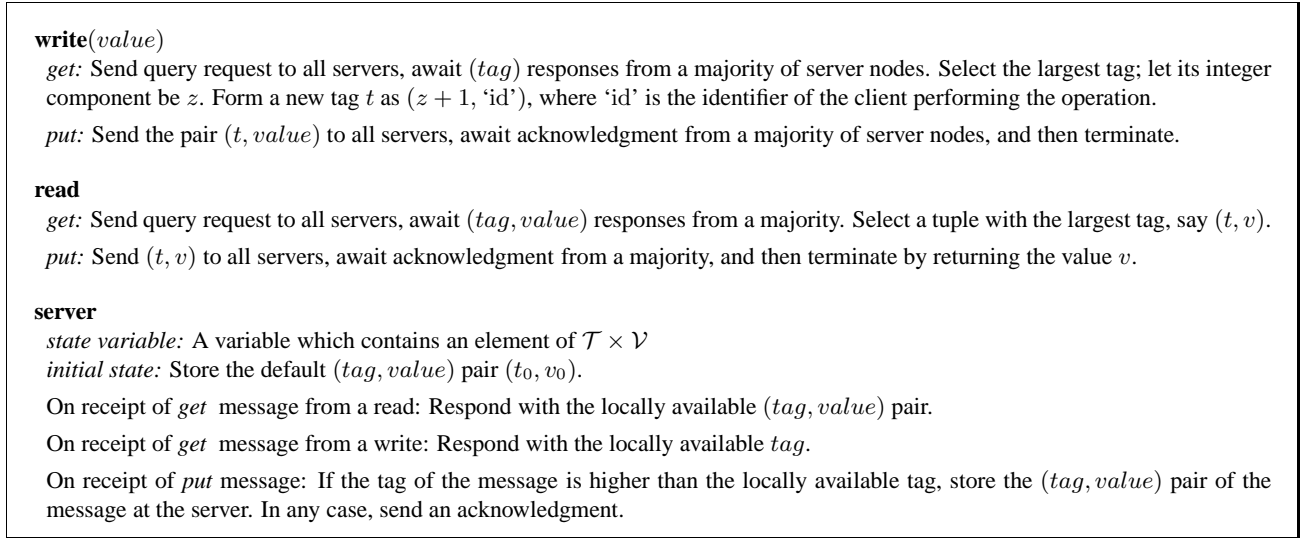

\begin{mdframed}[style=mymdframestyle]
\footnotesize
\begin{tabbing}
{\bf write$(value)$} \\
        \ \ \begin{minipage}[t]{\textwidth}%
{\em get:} {Send query request to all servers, await $(tag)$ responses from a majority of {server }nodes. 
    {Select the largest tag; let its integer component be $z$. Form a new tag $t$ as $(z+1,\text{`}\mathrm{id}\text{'})$, where `$\mathrm{id}$' is the identifier of the client performing the operation.}}\smallskip \\ 
{\em put:} Send {the pair }$(t,value)$ to all servers, await {acknowledgment} from a majority of {server }nodes, 
    and then terminate. 
    \end{minipage}\\ \\

{\bf read} \\
    \ \ \begin{minipage}[t]{\textwidth}%
{\em get:} Send query request to all servers, await $(tag,value)$ responses from a majority. Select 
a tuple with the largest tag, say $(t,v)$.\smallskip\\
{\em put:} Send $(t,v)$ to all servers, await {acknowledgment} from a majority, 
  and then terminate by returning the value $v$.
  \end{minipage}\\  \\ 
{\bf server} \\
    \ \ \begin{minipage}[t]{\textwidth}%
{{\em state variable:} A variable which contains an element of $\mathcal{T} \times \mathcal{V}$} \\
{{\em initial state:} Store the default $(tag, value)$ pair $(t_0,v_0)$.}\smallskip\\
{On receipt of {\em get } message from a read: Respond with the locally available $(tag, value)$ pair.}\smallskip \\
{ {On receipt of {\em get } message from a write:} Respond with the locally available $tag$.}\smallskip \\
{{On receipt of {\em put} message:} If the tag of the message is higher than {the} locally available tag, store the $(tag,value)$ pair of the message at the server. In any case, send an acknowledgment.}
\end{minipage}
\end{tabbing}
\end{mdframed}
\caption{Write, read{, and server} protocols of the ABD algorithm.}\label{fig:ABD}
\end{figure}

\begin{mainth}\label{claim:ABD}
{The write and read communication costs of ABD are respectively equal to $N \log|\mathcal{V}|$ and ${2N} \log|\mathcal{V}|$ bits. The storage cost is equal to $N$ \unit.}
\end{mainth}
The LDR algorithm divides its servers into \emph{directory servers} that store metadata, and \emph{replica servers} that store object values. The write protocol of LDR involves the sending of object values to $2f+1$ replica servers. The read protocol is less taxing since in the worst-case, it involves retrieving the data object values from $f+1$ replica servers. We state the communication costs of LDR next (for formal proof, see Appendix \ref{app:ABD_LDR}.)
\begin{mainth}\label{claim:comLDR}
In LDR, the write communication cost is $(2f+1)$ \unit, and the read communication cost is $(f+1)$ \unit.
\end{mainth}
In the LDR algorithm, each replica server stores every version of the data object it receives\footnote{This is unlike ABD where the servers store {only }the latest version of the data object received.}. Therefore, the (worst-case) storage cost of the LDR algorithm is unbounded. 

%

\begin{figure}[t]
\begin{mdframed}[style=mymdframestyle]
\footnotesize
\begin{tabbing}
{\bf write$(value)$} \\
        \ \ \begin{minipage}[t]{\textwidth}%
{\em get-metadata:} Send query request to directory servers, and await {$(tag, location)$} responses from a majority of directory servers. {Select the largest tag; let its integer component be $z$. Form a new tag $t$ as $(z+1,\text{`}\mathrm{id}\text{'})$, where `$\mathrm{id}$' represents the identifier of the client performing the operation.}\smallskip\\
{\em put:} Send $(t,value)$ to $2f+1$ replica servers, await {acknowledgment} from $f+1$.  Record identifiers of {the first $f+1$ replica} servers that respond, call this set of identifiers $\mathcal{S}$. \smallskip\\
{\em put-metadata:} Send $(t,\mathcal{S})$ to all directory servers, await {acknowledgment} from a majority, 
    and then terminate.
    \end{minipage}\\ \\
{\bf read} \\
    \ \ \begin{minipage}[t]{\textwidth}%
{\em get-metadata:} Send query request to directory servers, and await {($tag, location$)} responses from a majority {of directory servers.} Choose a {($tag, location$)} pair with the largest tag, let this pair be $(t,\mathcal{S}).$\smallskip\\
{\em put-metadata:} Send $(t,\mathcal{S})$ to all directory servers, await {acknowledgment} from a majority.\smallskip\\
{\em get:} Send \emph{get object} request to {any $f+1$} replica servers recorded in $\mathcal{S}$ for tag $t$. \str{A single response 
    suffices.} {Await a single response and }terminate by returning a value.
    \end{minipage}\\ \\
{{\bf replica server}  }\\
    \ \ \begin{minipage}[t]{\textwidth}%
{ {\em state variable:} A variable that is subset of $\mathcal{T}\times \mathcal{V}$} \smallskip\\
{{\em initial state:} Store the default $(tag, value)$ pair $(t_0,v_0)$.}\smallskip\\
{On receipt of {\em put } message: Add the $(tag,value)$ pair in the message to the set of locally available pairs. Send an acknowledgment.}\smallskip \\
{{On receipt of {\em get }message:} If the value associated with the requested tag is in the set of pairs stored locally, respond with the value. Otherwise ignore.}\smallskip 
\end{minipage}\\ \\
{\bf directory server} \\
    \ \ \begin{minipage}[t]{\textwidth}%
{\emph{state variable:} A variable that is an element of $\mathcal{T} \times 2^{\mathcal{R}}$ where $2^{\mathcal{R}}$ is the set of all subsets of $\mathcal{R}.$}\smallskip\\
{{\em initial state:} Store $(t_0, \mathcal{R}),$ where $\mathcal{R}$ is the set of all replica servers.}\smallskip\\
{On receipt of {\em get-metadata} message: Send the $(tag,\mathcal{S})$ {be the }pair stored locally. }\smallskip \\
{On receipt of {\em put-metadata} message:} Let $(t,\mathcal{S})$ be the incoming message. At the point of reception of the message, let $(tag, \mathcal{S}_1)$ {be the }pair stored locally at the server. If $t$ is equal to the $tag$ stored locally, then store $(t, \mathcal{S}\cup \mathcal{S}_{1})$ locally. If $t$ is bigger than $tag$ and if $|\mathcal{S}| \geq f+1,$ then store $(t,\mathcal{S})$ locally. {Send an acknowledgment.}
\end{minipage}
\end{tabbing}
\end{mdframed}
\caption{Write, read, {and server }protocols of the LDR algorithm}
\label{fig:LDR}
\end{figure}

{\noindent \bf Communication and Storage costs of ABD and LDR algorithms.}

\noindent \textit{Proof of Theorem \ref{claim:ABD}.}
    We first present arguments that upper bound the communication and storage cost for every execution of the ABD algorithm. The ABD algorithm presented here is fitted to our model.  Specifically in \cite{ABD,LS97} there is no clear cut separation between clients and servers.  However, {this} separation does not change the costs of the algorithm. {Then }we present worst-case executions that incur the costs as stated in the theorem.

    \noindent {\em Upper bounds:}
    First consider the write protocol. It has two phases, {\em get} and {\em put}. The \emph{get} phase of a write involves transfer of a tag, but {not of actual }data, and therefore has negligible communication cost. {In the \emph{put} phase of a write, the client sends {a value from the set} $\mathcal{T} \times \mathcal{V}$ to every server node;} the total communication cost of this phase is at most $N \log_2|\mathcal{V}|$ bits. Therefore the total {write }communication cost is at most $N \log_{2}|\mathcal{V}|$ bits. In the \emph{get} phase of the read protocol, the message from the client to the servers {contains} only metadata, and therefore has negligible communication cost. However, in this phase, each of the $N$ servers could respond to the client with a message from $\mathcal{T} \times \mathcal{V}$; therefore the total communication cost of the messages involved in the \emph{get} phase is upper bounded by $N \log_{2}|\mathcal{V}|$ bits. In the \emph{put} phase of the read protocol, the read sends an element {of} $\mathcal{T} \times \mathcal{V}$ to {$N$} servers. Therefore, this phase incurs a communication cost of at most $ {N} \log_{2}|\mathcal{V}|$ bits. The total communication cost of a read is therefore upper bounded by ${2N }\log_{2}|\mathcal{V}|$ bits.

    The storage cost of ABD is no bigger than $N\log_{2}|\mathcal{V}|$ bits because each {server} stores at most one value - the latest value it receives.

    \noindent {\em Worst-case executions:}
    Informally speaking, due to asynchrony and the possibility of failures, clients always send requests to all servers and in the worst case, all servers respond. Therefore the upper bounds described above are tight. 

    For the write protocol, the client sends the value to all $N$ nodes in its \emph{put} phase. So the {write communication }cost in an execution {where at least one write terminates }is $N \log_{2}|\mathcal{V}|$ bits. {For the read protocol, consider the following execution, where there is one read operation, and one write operation that is concurrent with this read. We will assume that none of the $N$ servers fail in this execution. Suppose that the writer completes its get phase, and commits to a tag $t$. Note that $t$ is the highest tag in the system at this point. Suppose that among the $N$ messages that the writer sends in its put phase with the value and tag $t$, Now the writer begins its put phase where it sends $N$ messages with the value and tag $t$. At least one of these messages, say the message to server $1$, arrives.the remaining messages are delayed, i.e., they are assumed to reach after the portion of the execution segment described here.  At this point, the read operation begins and receives $(tag, value)$ pairs from all the $N$ server nodes in its get phase. Of these $N$ messages, at least one message contains the tag $t$ and the corresponding value. Note that $t$ is the highest tag it receives. Therefore, the put phase of the read has to sends $N$ messages with the tag $t$ and the corresponding value - one message to each of the $N$ servers that which responded to the read in the get phase with an older tag. } 

{The read protocol has two phases. The cost of a read operation in an execution is the sum of the {communication} costs of the messages sent in its \emph{get} phase and those sent in its \emph{put} phase. The \emph{get} phase involves communication of $N$ messages from $\mathcal{T} \times \mathcal{V}$, one message from each server to the client, and therefore incurs a communication cost of $N \log_{2}|\mathcal{V}|$ bits {provided that} every server is active. The \emph{put} phase involves the communication of a message in $\mathcal{T} \times \mathcal{V}$ from {the} client to every server thereby {incurring} a communication cost of $N \log_{2}|\mathcal{V}|$ bits as well. Therefore, in any execution where all $N$ servers are active, the communication cost of a read operation is ${2N}\log_{2}|\mathcal{V}|$ bits and therefore the upper bound is tight.}

The storage cost is equal to $N \log_2|\mathcal{V}|$ bits since each of the $N$ servers store exactly one value from $\mathcal{V}$. \qed

\noindent \textit{Proof of Theorem \ref{claim:comLDR}.}


\noindent {{\emph{Upper bounds:}}
    In LDR servers are divided into two groups: {\em directory} servers used to manage object metadata, and {\em replication} servers used for object replication. Read and write protocols have three sequentially executed phases. The {\em get-metadata} and {\em put-metadata} phases incur negligible communication cost since only metadata is sent over the message-passing system.  In the {\em put} phase, the writer sends its messages, each of which is an element from $\mathcal{T} \times \mathcal{V},$ to $2f+1$ replica servers and awaits $f+1$ responses; since the responses have negligible communication cost, this phase incurs a total communication cost of at most $(2f+1)\log_2|\mathcal{V}|$ bits. The read protocol is less taxing, where the reader during the {\em get} phase queries $f+1$ replica servers and in the worst case, all respond with a message containing an element from $\mathcal{T} \times \mathcal{V}$ thereby incurring a total communication cost of at most $(f+1)\log_2|\mathcal{V}|$ bits.

    {\noindent \emph{Worst-case executions:}}
    It is clear that in every execution {where at least one writer terminates,} the writer sends out $(2f+1)$ messages to replica servers that contain the value, {thus incurring a write communication cost of $(2f+1)\log_{2}|\mathcal{V}|$ bits.} Similarly, for a read, in certain executions, all $(f+1)$ replica servers {that are selected in the \emph{put phase} of the read} respond to the $\emph{get}$ request from the client. So the upper bounds derived above are tight. \qed
}

\section{Discussion on Erasure Codes}\label{app:coding}
For an $(N,k)$ code, the ratio $\frac{N}{k}$ - also known as the \emph{redundancy factor} of the code - represents the storage cost overhead in the classical erasure coding model. Much literature in coding theory involves the design of $(N,k)$ codes for which the redundancy factor\footnote{Literature in coding theory literature often studies the \emph{rate} $\frac{N}{k}$  of a code, which is the reciprocal of the redundancy factor, i.e., the rate of an $(N,k)$ code is $\frac{k}{N}.$ In this paper, we use the redundancy factor in our discussions since it enables a somewhat more intuitive connection with the costs of our algorithms in Theorems \ref{claim:ABD}, \ref{claim:comLDR}, \ref{thm:comCAS}, \ref{thm:storCAS}.} can be made as small as possible.
In the classical erasure coding model, the extent to which the redundancy factor can be reduced depends on $f$ - the maximum number of server failures that are to be tolerated. In particular, an $(N,k)$ MDS code, when employed to store the value of the data object, tolerates $N-k$ server node failures; this is because the definition of an MDS code implies that the data can be recovered from any $k$ surviving nodes. Thus, for an $N$-server system {that uses an MDS code}, we must have $k \leq N-f$, meaning that the redundancy factor is at least $\frac{N}{N-f}$. It is well known \cite{Roth_CodingTheory} that, given $N$ and $f$, {the parameter }$k$ cannot be made larger than $N-f$ so that the redundancy factor is lower bounded by $\frac{N}{N-f}$ for \emph{any} code{ even if it is not an MDS code; In fact, an MDS code can equivalently be defined as one which attains this lower bound on the redundancy factor. In coding theory, this lower bound is known as the Singleton bound \cite{Roth_CodingTheory}. Given parameters $N,k,$ the question of whether an $(N,k)$ MDS code exists depends on the alphabet of code $\mathcal{W}$. We next discuss some of the relevant assumptions that we (implicitly) make in this paper to enable the use of an $(N,k)$ MDS code in our algorithms.

\subsection*{Assumption on $|\mathcal{V}|$ due to Erasure Coding}
{Recall that, in our model, each value $v$ of a data object belongs to a finite set $\mathcal{V}$. In our system, for the use of coding, we assume that $\mathcal{V}=\mathcal{W}^{k}$ for some finite set $\mathcal{W}$ and that $\Phi:\mathcal{W}^{k} \rightarrow \mathcal{W}^{N}$ is an MDS code. Here we refine these assumptions using classical results from erasure coding theory. In particular, the following result is useful.}
\begin{mainth}
{Consider a finite set $\mathcal{W}$ such that $|\mathcal{W}| \geq {N}.$ Then, for any integer $k < N$, there exists an $(N,k)$ MDS code $\Phi:\mathcal{W}^{k} \rightarrow \mathcal{W}^{N}$.}
\end{mainth}  
{One proof for the above in coding theory literature is constructive. Specifically, it is well known that when $|\mathcal{W}| \geq {N}$, then $\Phi$ can be constructed using the Reed-Solomon code construction \cite{Reed_Solomon,Roth_CodingTheory,LinCostello_Book}. The above theorem implies that, to employ a Reed-Solomon code over our system, we {shall} need the following two assumptions:}
{\begin{itemize}
    \item  $k$ divides $\log_2|\mathcal{V}|,$ and
        \item $\log_{2}|\mathcal{V}|/k ~\geq ~ \log_{2}N$.
        \end{itemize}
    Thus all our results are applicable under the above assumptions. }

{In fact, the first assumption above can be replaced by a different assumption with only a negligible effect on the communication and storage costs. Specifically, if $\log_2|\mathcal{V}|$ were not a multiple of $k$ then, one could pad the value with $\left(\lceil\frac{\log_2{|\mathcal{V}}|}{k}\rceil k - \log_2|\mathcal{V}|\right)$ ``dummy'' bits, all set to 0, to ensure that the (padded) object has a size that is multiple of $k$; note that this padding is an overhead. The size of the padded object would be $\lceil\frac{\log_2{|\mathcal{V}|}}{k}\rceil k$ bits and the size of each coded element would be $\lceil\frac{\log_2{|\mathcal{V}|}}{k}\rceil$ bits. If we assume that $\log_{2}|\mathcal{V}| \gg k$ then, $\lceil\frac{\log_2{|\mathcal{V}|}}{k}\rceil \approx \frac{\log_{2}|\mathcal{V}|}{k}$ meaning that the padding overhead can be neglected. Consequently, the first assumption can be replaced by the assumption that $\log_{2}|\mathcal{V}| \gg k$ with only a negligible effect {on} the communication and storage costs.}

\section{Proof of Lemma \ref{lem:quorums} }\label{appendix:sec:quorum}

\noindent Proof of property ({\bf i}):
	By the definition, each $Q \in \mathcal{Q}$ has cardinality at least $\lceil \frac{N+k}{2}\rceil$. Therefore, for $Q_1, Q_2 \in \mathcal{Q},$ we have
\begin{eqnarray*}
|Q_{1} \cap Q_{2}| &=& |Q_1|+|Q_2|-|Q_1 \cup Q_2|\\
& {\geq} &2 \left\lceil \frac{N+k}{2}\right\rceil - |Q_1 \cup Q_2| \\
& \stackrel{{(a)}}{\geq} &2 \left\lceil \frac{N+k}{2}\right\rceil - N ~~{\geq}~~ k,
\end{eqnarray*}
where we have used the fact that $|Q_1 \cup Q_2| \leq N$ in {$(a)$}.

\noindent Proof of property ({\bf ii}): 
{Let $\mathcal{B}$ be the set of {all the server nodes that fail in an execution}, where $|\mathcal{B}| \leq f$. We need to show that there exists at least one quorum set $Q \in \mathcal{Q}$ such that $Q \subseteq \mathcal{N}-\mathcal{B}$, that is, at least one quorum survives. To show this, because of the definition of our quorum system, it suffices to show that $|\mathcal{N}-\mathcal{B}| \geq \lceil \frac{N+k}{2}\rceil$. We show this as follows:
\begin{eqnarray*}
	|\mathcal{N}-\mathcal{B}| &\geq& N-f ~\stackrel{{(b)}}{\geq} ~ N-\left\lfloor \frac{N-k}{2}\right\rfloor ~= ~\left\lceil\frac{N+k}{2} \right\rceil,
\end{eqnarray*}
where, ${(b)}$ follows because $k \leq N-2f$ implies that $f \leq \lfloor\frac{N-k}{2}\rfloor$.}

\end{document}